\documentclass{amsproc}

\usepackage{bm} 
\usepackage{url}
\usepackage{caption}
\usepackage{subcaption}
\usepackage{thmtools}
\usepackage{thm-restate}
\usepackage{mathtools}

\usepackage[table]{xcolor}
\usepackage{colortbl}

\usepackage{hyperref}
\usepackage[colorinlistoftodos]{todonotes}

\usepackage{xspace}
\usepackage{algorithmicx}
\usepackage{algorithm}
\usepackage{algpseudocode}
\usepackage[shortlabels]{enumitem}

\usepackage{tikz}
\usetikzlibrary{calc}
\usetikzlibrary{shapes,decorations.markings,arrows.meta,fit}

\usepackage[normalem]{ulem}

\usepackage{microtype}

\usepackage{qtree}

\usepackage{relsize}

\newcommand{\rai}{RelationalAI}

\newcommand{\norm}[1]{\|#1\|}
\newcommand{\set}[1]{\{#1\}}                    
\newcommand{\dom}{\textsf{Dom}}

\newcommand{\polylog}{\text{\sf polylog}}

\newcommand{\opt}{\text{\sf OPT}}

\newcommand{\calS}{\mathcal S}

\newcommand{\calB}{\mathcal B}

\newcommand{\calD}{\mathcal D}
\newcommand{\calL}{\mathcal L}

\newcommand{\calT}{\mathcal T}

\newtheorem{theorem}{Theorem}[section]
\newtheorem{lemma}[theorem]{Lemma}

\theoremstyle{definition}              

\newtheorem{thm}[theorem]{Theorem}

\newcommand{\defeq}{\stackrel{\text{def}}{=}}

\newcommand{\N}{\mathbb N} 
\newcommand{\R}{\mathbb R} 

\newcommand{\cd}{\text{ :- }}

\newcommand{\td}{\text{\sf TD}}

\newcommand{\bags}{\text{\sf bags}}

\newcommand{\OUT}{\text{\sf OUT}}

\newcommand{\atoms}{\text{\sf atoms}}

\newcommand{\bs}{\text{\sf BS}}

\newcommand{\panda}{\textsf{PANDA}\xspace}

\newcommand{\pandaexpress}{\textsf{PANDAExpress}\xspace}

\newcommand{\fhtw}{\textsf{fhtw}}
\newcommand{\subw}{\textsf{subw}}
\newcommand{\cost}{\textsf{cost}}

\newcommand{\ov}{\overline}

\newcommand{\positiveterm}[1]{\fcolorbox{black}{red!35}{#1}}
\newcommand{\negativeterm}[1]{\colorbox{blue!35}{#1}}

\newcommand{\scq}{{\sf{\#CQ}}\xspace}
\newcommand{\scqs}{{\sf{\#CQs}}\xspace}
\newcommand{\faq}{{\sf{FAQ}}\xspace}
\newcommand{\faqs}{{\sf{FAQs}}\xspace}
\newcommand{\cycle}{Q_\square}
\newcommand{\cyclestats}{\calS_\square}
\newcommand{\cyclefull}{Q_\square^{\mathsf{full}}}
\newcommand{\cyclefullstats}{\calS_\square^{\mathsf{full}}}
\newcommand{\cyclebool}{Q_\square^{\mathsf{bool}}}

\newcommand{\mm}{\mathsf{MM}}
\newcommand{\osubw}{\textsf{$\omega$-subw}}

\newcommand{\entw}{\textsf{entw}}

\newcommand{\np}{\mathrm{NP}}
\newcommand{\ci}{\mathrm{CI}}

\colorlet{lightgray}{gray!50!white}
\colorlet{lightergray}{gray!20!white}

\usepackage{fullpage}
\usepackage{amssymb}
\newcommand{\myemail}[1]{\small\texttt{#1}}

\allowdisplaybreaks

\begin{document}

\title{Query Optimization and Evaluation via Information Theory: A Tutorial}
\author{
  \begin{tabular}{c}
    Mahmoud Abo Khamis\\
    \rai\\
    \myemail{mahmoud.abokhamis@relational.ai}
  \end{tabular}
  \quad
  \begin{tabular}{c}
    Hung Q. Ngo\\
    \rai\\
    \myemail{hung.q.ngo@gmail.com}
  \end{tabular}
  \quad
  \begin{tabular}{c}
    Dan Suciu\\
    University of Washington\\
    \myemail{suciu@cs.washington.edu}
  \end{tabular}
}

\begin{abstract}
Database theory is exciting because it studies highly general and practically useful
abstractions. Conjunctive query (CQ) evaluation is a prime example: it simultaneously
generalizes graph pattern matching, constraint satisfaction, and statistical inference,
among others. This generality is both the strength and the central challenge of the field.
The query optimization and evaluation problem is fundamentally a \emph{meta-algorithm}
problem: given a query $Q$ and statistics $\calS$ about the input database, how should one
best answer $Q$? Because the problem is so general, it is often impossible for such a
meta-algorithm to match the runtimes of specialized algorithms designed for a fixed
query---or so it seemed.

The past fifteen years have witnessed an exciting development in database
theory: a general framework, called $\panda$, that emerged from advances in
database theory, constraint satisfaction problems (CSP), and graph algorithms,
for evaluating conjunctive queries given input data statistics. The key idea is to derive information-theoretically
tight upper bounds on the cardinalities of intermediate relations produced during
query evaluation. These bounds determine the costs of query plans, and
crucially, the query plans themselves are derived directly from the mathematical
\emph{proof} of the upper bound. This tight coupling of proof and algorithm is
what makes $\panda$ both principled and powerful.

Remarkably, this generic algorithm matches---and in some cases subsumes---the
runtimes of specialized algorithms for the same problems, including algorithms
that exploit fast matrix multiplication.
This paper is a tutorial on the $\panda$ framework. We illustrate the key ideas
through concrete examples, conveying the main intuitions behind the theory.
\end{abstract}


\maketitle

\section{Introduction}
\label{sec:intro}

Conjunctive queries (CQs) are a remarkably versatile syntactic abstraction that arises
naturally across many areas of computer science, including database query
evaluation~\cite{DBLP:books/aw/AbiteboulHV95,DBLP:conf/stoc/ChandraM77}, constraint
satisfaction~\cite{DBLP:conf/pods/KolaitisV98,DBLP:journals/jacm/Marx13,DBLP:journals/talg/GroheM14}, machine learning
and tensor computation~\cite{DBLP:conf/pods/KhamisNR16}, and
logic~\cite{DBLP:conf/stoc/ChandraM77}. A CQ is essentially a join query: it asks for all
combinations of tuples, one from each relation, that agree on their shared variables, and
then projects the result onto a chosen subset of variables. For example, the following are
two CQs over four relations $R, S, T, U$:
\begin{align}
    \cyclefull(X, Y, Z, W) &\quad\cd\quad R(X, Y) \wedge S(Y, Z) \wedge T(Z, W) \wedge U(W, X)
    \label{eq:4cycle:full}\\
    \cycle(X, Y) &\quad\cd\quad R(X, Y) \wedge S(Y, Z) \wedge T(Z, W) \wedge U(W, X)
    \label{eq:4cycle}
\end{align}
Given a {\em database instance $\calD$} with {\em relation instances} (i.e.~instantiations
of the tables) $R, S, T, U$, the first query $\cyclefull$ asks for all tuples $(X, Y, Z, W)$
such that the conjunction $R(X, Y) \wedge S(Y, Z) \wedge T(Z, W) \wedge U(W, X)$ holds true
in $\calD$. The second query $\cycle$ asks for all tuples $(X, Y)$ for which there exist $Z,
W$ making the same conjunction hold. Figure~\ref{fig:4cycle:data} shows an example input
database instance along with the output of $\cyclefull$ over this instance.

This syntactic abstraction is powerful because it captures a wide range of problems. For
example, if we view the relations $R, S, T, U$ as edges in a graph, then $\cyclefull$ asks
for all 4-cycles in the graph, while $\cycle$ asks for all edges that can be extended to
form a 4-cycle. In fact, any graph pattern matching problem can be expressed as a CQ
evaluation problem. Moreover, CQs also capture constraint satisfaction problems (CSPs) where
the relations represent constraints on variable assignments. By changing the semiring from
Boolean to another semiring, CQs can also capture certain statistical inference problems
where the relations represent probabilistic dependencies~\cite{provenance-semirings,DBLP:conf/pods/GreenT17,DBLP:conf/pods/KhamisNR16}.

Despite this common abstraction, different communities have developed their own ways of
parameterizing the input and measuring algorithmic efficiency. In graph pattern matching,
complexity is typically expressed in terms of the number of vertices and edges. In tensor
and matrix computation, it is the dimensionality of the tensors. In databases, it is various
statistics collected from the input relations, such as relation sizes, the number of
distinct values per column, or functional dependencies. Each community has also developed
its own algorithmic toolkit. The graph algorithm community, in particular, has produced
efficient algorithms tailored to specific queries or small classes of
queries~\cite{DBLP:conf/stoc/Itai77, DBLP:journals/algorithmica/AlonYZ97,
DBLP:journals/siamcomp/DalirrooyfardVW21, 10.1145/3618260.3649663,
DBLP:conf/stoc/BringmannG25}. The CSP community has devised algorithms based on {\em tree
decompositions}~\cite{GOTTLOB2002579,DBLP:journals/talg/GroheM14,DBLP:journals/jacm/Marx13},
which exploit the structural properties of the query hypergraph to decompose the problem
into smaller subproblems. The database community, on the other hand, relies on {\em query
plans} --- strategies for ordering and combining join
operations~\cite{DBLP:books/aw/AbiteboulHV95,Maier:1983:TRD:1097039,DBLP:books/daglib/0011128}.
We review these developments in some details in Section~\ref{sec:related}. All in all, there
is a surprisingly diverse set of input parameters, algorithm representations, and analytical
frameworks for studying algorithmic efficiency, despite the fact that these communities are
solving the same underlying problem: ``Given inputs satisfying certain statistics, how can
we best evaluate a CQ?''

There is, however, a fundamental difference between the database setting and those of other
communities. In databases, the query optimizer does not know the query or the statistics in
advance: it is a {\em meta-algorithm} that, upon being given a query $Q$ and statistics
$\calS$, produces an efficient evaluation algorithm on the fly. This is in stark contrast
with the typical algorithmic setting, where the query (or a small class of queries, such as
$k$-cycles for a fixed $k$) is fixed ahead of time and the algorithm is designed
specifically for it. This distinction makes it especially important from the database
viewpoint to develop a {\em unified abstraction} --- one that can represent arbitrary input
statistics $\calS$ and arbitrary queries $Q$, reason about the complexity of evaluating $Q$
under $\calS$ in a query- and statistics-agnostic way, and at the same time describe the
class of algorithms that the optimizer can output, so that we can both bound the best
achievable complexity and construct an algorithm that meets it.

It is thus quite remarkable that, over the past 15 years or so, a unified query optimization
and evaluation framework --- called
$\panda$~\cite{DBLP:conf/pods/Khamis0S17,theoretics:13722,pandaexpress} --- has {\em emerged
from and built upon advances} across the graph algorithms, CSP, and database communities.
This framework models a wide class of input statistics $\calS$, provides a general
representation for the class of algorithms that the optimizer can output, and yields tight
complexity analyses --- all in a query- and statistics-agnostic manner. Moreover, it
subsumes and explains much of what was previously known from each of these communities as
special cases.

The $\panda$ framework has four key components: input statistics, query plan representation,
cost estimation, and query plan search. We briefly describe these components here, and will
elaborate on them in the subsequent sections.

$\panda$ begins with an abstraction for representing input statistics, called {\em degree
constraints}~\cite{DBLP:conf/pods/KhamisNS16}. At a high level, a degree constraint on an
input relation bounds the number of distinct tuples over some subset of columns, for any
fixed assignment to another subset of columns. For example, an upper bound on the size of a
relation is a degree constraint (called a {\em cardinality constraint}); so is a {\em
functional dependency} (FD) $W \to X$ on a relation $U(W, X)$, which asserts that each
$W$-value determines at most one $X$-value; and so is an upper bound $C$ on the {\em maximum
degree}, i.e.~the number of $X$-values associated with any fixed $W$-value in $U(W, X)$.
This is a surprisingly general and robust notion: in graph pattern matching, it captures the
number of edges or vertices, and outdegrees; in tensor and matrix computation, it captures dimensionality and sparsity;
and in database systems, it subsumes common statistics such as relation sizes, the number of
distinct values per column, and functional dependencies. There are also other more exotic
forms of statistics called $\ell_p$-norm constraints that fit naturally into this
framework~\cite{2021arXiv211201003V,10.1145/3651597}.

The second component of $\panda$ is a representation of query plans. {\em Tree
decompositions} (TDs), which grew out of the
CSP~\cite{GOTTLOB2002579,DBLP:journals/talg/GroheM14} and probabilistic graphical
models~\cite{DBLP:books/daglib/0023091} communities, provide a general and principled
representation of query plans, subsuming join trees in
databases~\cite{DBLP:conf/vldb/Yannakakis81,DBLP:conf/pods/GottlobGLS16}. Roughly speaking, a TD of a CQ represents a
structured network of joins and projections that computes the query, where the structure is
captured by a tree. We refer to query plans based on a single TD a ``static'' query plan.
The query plan representation in $\panda$ is not static, however. Drawing from a brilliant
idea of Marx~\cite{DBLP:journals/jacm/Marx13}, a $\panda$ query plan involves {\em multiple}
tree decompositions, with the data appropriately partitioned to load-balance the work across
them. These are also called ``dynamic'' query plans. Somewhat surprisingly, this richer
representation is powerful enough to capture state-of-the-art results from graph pattern
matching such as those
in~\cite{DBLP:journals/algorithmica/AlonYZ97,DBLP:conf/stoc/BringmannG25}. Having multiple
join trees in a query plan is a highly novel idea from a database theory and systems
perspective, and as we shall see, it is crucial for achieving the best possible runtimes for
certain queries. A major building block of dynamic query plans is the ability to evaluate
{\em disjunctive datalog rules} (DDRs) --- datalog rules with disjunctions in the head ---
which are more general than CQs.

Having defined the search space of query plans, the next challenge is to {\em cost} a plan
from this space and {\em search} for the best one --- a process that is quintessentially
database management, and is known as the {\em query optimization
problem}~\cite{DBLP:conf/pods/Chaudhuri98,DBLP:journals/ftdb/DingNC24}. Roughly, our cost function measures the largest
cardinality over all intermediate relations computed during the execution of a query plan.
However, cardinality estimation is notoriously the ``Achilles' heel of query
optimizers''~\cite{DBLP:journals/pvldb/LeisGMBK015}: it is a difficult problem that has
always been a challenge despite decades of research and development. Fortunately,
information theory comes to the rescue. Drawing from a long line of work on using entropic
arguments to bound the number of occurrences of subgraph
patterns~\cite{MR599482,MR859293,MR1639767,DBLP:journals/siamcomp/AtseriasGM13}, our
cardinality estimate for an intermediate relation is the worst-case size of that relation,
over all databases satisfying the input degree constraints, when computed using the given
query plan. Computing this estimate is itself an optimization problem over the so-called
{\em entropic functions} --- functions that arise as joint entropy vectors of collections of
random variables. Since characterizing entropic functions is a major open problem in
information theory~\cite{kolaitis_et_al:DagRep.12.7.180}, we relax this optimization to one
over the larger class of {\em polymatroids}, subject to a particular {\em Shannon
inequality} constraint. This problem is both tractable and yields provably tight bounds in
many cases of interest.

Now that we know how to estimate the cost of a query plan, the optimizer must search for the
optimal one. This is where another key idea in $\panda$ comes into play. Rather than
searching over query plans independently of the cost estimation, $\panda$ exploits the {\em
symbolic proof} of the Shannon inequality constraint used in the cardinality estimation
step, turning each proof step into an atomic relational operator --- aggregation, partition,
projection, or join --- which together form the query plan whose cost is exactly what the
optimization problem predicts. The idea of turning a mathematical proof of a cardinality
bound into an algorithm was pioneered in the design of {\em worst-case optimal join
algorithms}~\cite{DBLP:conf/pods/NgoPRR12,DBLP:journals/sigmod/NgoRR13,10.1145/3180143}.
However, $\panda$ exploits a fundamentally different bound, one based on {\em
polymatroids} and {\em sub-probability measures}, which enables it to handle the richer
query plan representation involving multiple tree decompositions.

Putting it all together, $\panda$ can evaluate any conjunctive query $Q$ on a database
instance satisfying degree constraints $\calS$ in time $O(N^{\subw(Q,\calS)} \cdot \log N +
\OUT)$, where $N$ is the input size, $\OUT$ is the output size, and $\subw(Q,\calS)$ is a
parameter --- called the {\em submodular width} --- that depends only on the query and the
statistics. Remarkably, this runtime simultaneously matches several state-of-the-art
results: it recovers the best known combinatorial runtimes for graph pattern matching,
including those in~\cite{DBLP:journals/algorithmica/AlonYZ97,DBLP:conf/stoc/BringmannG25};
and it achieves the submodular width runtime predicted by Marx's framework for tractable
CSPs~\cite{DBLP:journals/jacm/Marx13}. By adding matrix multiplication as an additional
operator to the query plan, $\panda$ can further match algorithms based on fast matrix
multiplication~\cite{DBLP:journals/pacmmod/KhamisHS25}. It also extends naturally to handle
$\ell_p$-norm constraints, which strictly generalize degree
constraints~\cite{2021arXiv211201003V,10.1145/3651597}. Crucially,
$\panda$ achieves this generality for arbitrary conjunctive queries --- not just Boolean or
full CQs.

This paper aims to be a tutorial on this thrilling development. We shall focus on examples
that illustrate the key ideas, referring the readers to
the original
PANDA papers~\cite{DBLP:conf/pods/Khamis0S17,
theoretics:13722,
pandaexpress,
DBLP:journals/pacmmod/KhamisHS25}
for more details and technicalities. We hope that this tutorial will make the $\panda$ framework
more accessible to a wider audience, and inspire further research on this exciting topic.


\paragraph*{\bf Paper Organization}
In Section~\ref{sec:related}, we review more related work on query evaluation as well as
lower bounds. In Section~\ref{sec:prelims}, we present necessary preliminaries on
statistics, information theory, and tree decompositions.
Section~\ref{sec:static:query:plans} presents information-theoretic cost estimate bounds
for static (single-TD) query plans. Section~\ref{sec:adaptive:query:plans} extends this to
adaptive (multi-TD) query plans, introduces DDRs, and derives the submodular width as the
key complexity measure. Section~\ref{sec:subw} explains how to compute the submodular width
via linear programming, and how {\em Shannon-flow inequalities} arise as the dual
certificates. Before we describe the $\panda$ algorithm for evaluating DDRs in
Section~\ref{sec:panda}, we present in Section~\ref{sec:proof:sequences} the notion of proof
sequences for establishing Shannon-flow inequalities. This completes the description of the
end-to-end $\panda$ algorithm for CQ evaluation.
Section~\ref{sec:extensions} discusses
various extensions of $\panda$ to handle other classes of queries, statistics, as well as
FMM. We conclude in Section~\ref{sec:open-problems} with some open problems.

\section{Related Work}
\label{sec:related}

\subsection{Query evaluation algorithms}
\label{sec:related:algorithms}

A special case of conjunctive query evaluation is the problem of finding a small
(hyper)graph pattern inside a large graph. A long line of work on this
problem~\cite{DBLP:conf/stoc/Itai77,MR599482,
MR859293,MR1639767,MR2104047,DBLP:conf/soda/GroheM06,
DBLP:journals/talg/GroheM14,DBLP:journals/siamcomp/AtseriasGM13} lead to the {\em
AGM-bound}~\cite{DBLP:journals/siamcomp/AtseriasGM13} for the worst-case output size of
conjunctive queries. The AGM-bound assumes that we know the sizes of the input relations.
Later works generalized and derived tighter bounds by incorporating more statistics about
the input such as functional dependencies~\cite{DBLP:journals/jacm/GottlobLVV12}, degree
constraints~\cite{DBLP:conf/pods/KhamisNS16}, and $\ell_p$-norms of the degree
vectors~\cite{10.1145/3651597}. The most general computable bound is the {\em polymatroid
bound}~\cite{DBLP:conf/pods/Khamis0S17}, which was extended to deal with $\ell_p$-norms
in~\cite{10.1145/3651597}.

Proofs of these bounds have led to novel join algorithms. The AGM-bound was shown
in~\cite{DBLP:conf/pods/NgoPRR12} to be equivalent to an isoperimetric inequality  by
Bollob\'as and Thomason~\cite{MR2104047}. Bollob\'as-Thomason's inductive proof of the
inequality suggests an inductive proof of the AGM-bound via the query-decomposition
lemma~\cite{DBLP:journals/sigmod/NgoRR13,DBLP:conf/pods/NgoPRR12}. This proof structure is
the basis of the first wave of worst-case optimal join
algorithms~\cite{DBLP:journals/sigmod/NgoRR13,DBLP:conf/pods/NgoPRR12,DBLP:conf/icdt/Veldhuizen14}.
These are algorithms running in time proportional to the AGM-bound.

When adding more statistics to tighten the output size
bound~\cite{DBLP:journals/jacm/GottlobLVV12,DBLP:conf/pods/KhamisNS16}, we needed a
different proof technique based on information theory. The entropy argument pioneered
in~\cite{MR859293} was adapted
in~\cite{DBLP:journals/siamcomp/AtseriasGM13,DBLP:journals/jacm/GottlobLVV12,DBLP:conf/pods/KhamisNS16}
to prove the AGM-bound and its generalizations. The most general bound dealing with
arbitrary degree constraints, called the {\em polymatroid bound}, was proved
in~\cite{DBLP:conf/pods/KhamisNS16}, which was also the first paper to introduce the idea of
turning a proof of the polymatroid bound into an algorithm. Three algorithms were proposed
(CA, SMA, and CSMA) based on different ``proof sequences'' of the polymatroid bound with
increasing power. Each step of a proof sequence is interpreted as a symbolic instruction to
guide the join algorithm. This is the first known instance of query plans devised from
information theoretic proofs. The types of proof sequences introduced were not complete,
however.

The bounds and results mentioned above apply to full conjunctive queries, where worst-case
optimality makes sense. This concept breaks down when the query has projections. As an
extreme case consider Boolean conjunctive queries where the output size is $O(1)$, there is
no hope of having a worst-case optimal algorithm. To deal with this issue, the constraint
satisfaction and machine learning communities have introduced query plans based on a
tree-decomposition of the input query since the
90s~\cite{zhang1994simple,DBLP:journals/ai/Dechter99}, where dynamic programming helps
answer a Boolean CQ in time $\tilde O(N^w)$ where $w$ is the ``width'' of a query.
There are
queries where the output size is very large (e.g. $N^n$) while the width is small (e.g.
$w=1$). These widths evolved from tree width\footnote{For tree width, the runtime is stated as $\tilde O(N^{w+1})$.}, (generalized) hypertree
width~\cite{DBLP:conf/pods/GottlobGLS16}, fractional hypertree
width~\cite{DBLP:conf/soda/GroheM06}, to ultimately the submodular
width~\cite{DBLP:journals/jacm/Marx13}.

While earlier width notions were based on query plans constructed from a single
tree-decomposition, the {\em submodular width} from Marx~\cite{DBLP:journals/jacm/Marx13}
was based on a beautiful idea that having multiple tree-decompositions to distribute the
computational load is crucial in reducing the runtime for answering Boolean CQs. Let
$\subw(Q)$ denote the submodular width of a Boolean query $Q$, then Marx's algorithm runs in
time $\tilde O(N^{c \cdot \subw(Q)})$ to answer $Q$ on any database of size $N$, and $c>1$
is a constant. Marx's algorithm does not deal with degree constraints and arbitrary free
variables.

The notion of submodular width is very powerful. In addition to it being optimal from a
tractability of CSP point of view~\cite{DBLP:journals/jacm/Marx13}, it also is the precise
{\em optimal} parameter for a class of sub-graph identification and listing queries in the
fined-grained complexity setting. Very recently, Bringmann and
Gorbachev~\cite{DBLP:conf/stoc/BringmannG25}, showed that $O(m^{\subw(H)}\log m)$ time is
the optimal time to find a subgraph pattern $H$ in a graph with $m$ edges, given that $H$ is
a sub-quadratic pattern, based on standard fine-grained complexity hypotheses.

Other forms of data partitioning have been used to develop {\em data-sensitive}
algorithms for (cyclic) CQs~\cite{deeds_et_al:LIPIcs.ICDT.2025.17,10.1145/2967101}.
Recently, data partitioning has also been used to develop {\em output-sensitive}
algorithms for acyclic CQs~\cite{10.1145/3695838,10.1145/3725241}
that improve upon the classic Yannakakis algorithm~\cite{DBLP:conf/vldb/Yannakakis81}.

\subsection{Lower bounds}
\label{sec:related:lower-bounds}

As mentioned before, the submodular width runtime that is achieved by the $\panda$ algorithm
remains the best known runtime, as far as combinatorial algorithms can go, for any CQ $Q$
that has been studied in the literature. This raises the questions: {\em Does this follow
from some existing lower bound conjecture in fine-grained complexity, or does this need its
own conjecture?} We state here some lower bound results that might be relevant to this line
of investigation.

Marx~\cite{DBLP:journals/jacm/Marx13} showed that {\em bounded submodular width} is indeed
the right criteria that captures precisely tractable Boolean CQs (without self joins), under
the Exponential Time Hypothesis (ETH). This however does not say anything about the {\em
fine-grained complexity} of Boolean CQs with bounded submodular width: It only says that
they are (the only ones) solvable in polynomial time.

Fan, Koutris, and Zhao~\cite{fan_et_al:LIPIcs.ICALP.2023.127} relied on the hardness of
$k$-clique queries under combinatorial algorithms to prove conditional hardness of all CQs.
The resulting lower bound measure, called the {\em clique-embedding power}, is always upper
bounded by the submodular width. It coincides with submodular width for some classes of
queries but deviates for others. The authors left open the question of whether this gap is
due to the non-tightness of their lower bound or the submodular width as an upper bound.

Bringmann and Gorbachev~\cite{DBLP:conf/stoc/BringmannG25} studied graph queries (i.e.~ with
binary input relations) in the region where the submodular width is strictly below 2.
They show that in this region, the submodular width does indeed provide a lower bound for
each one of the {\em listing}, {\em enumeration}, and {\em minimum-weight} versions of the problem, modulo well-accepted conjectures in fine-grained complexity.

A new kind of unconditional lower bounds emerged on the size of the smallest {\em circuit} needed to
represent the output of a CQ. Fan, Koutris, and Zhao~\cite{10.1145/3651588} showed that the
{\em entropic version}\footnote{The entropic version of the submodular width is the one
where polymatroids are replaced with entropic functions. We discuss this further in
Section~\ref{sec:open-problems}.}
of the submodular width, does indeed provide a lower bound
on the circuit size that cannot be improved by any polynomial factor. More recently,
Berkholz and Vinall-Smeeth~\cite{2025arXiv250320438B} showed a lower bound of
$N^{\Omega(\subw^{1/4})}$ on the circuit size.

\section{Preliminaries}
\label{sec:prelims}

We fix a set of variables $\bm V$ and a domain $\dom$. Variables are denoted by capital
letters ($X$) and their values by small letters ($x$); sets of variables by boldface
capitals ($\bm X$) and tuples of values by boldface small letters ($\bm x$). We write
$\dom^{\bm X}$ for the set of all tuples over $\bm X$, and $\bm x_{\bm Y}$ for the
projection of $\bm x \in \dom^{\bm X}$ onto $\bm Y \subseteq \bm X$.

\subsection{Conjunctive Queries}
A {\em conjunctive query} (CQ) $Q$ is an expression of the form:
\begin{align}
    Q(\bm F) \quad\cd\quad \bigwedge_{R(\bm X)\in \atoms(Q)} R(\bm X),
    \label{eq:cq}
\end{align}
where $\atoms(Q)$ is the {\em set of atoms} in $Q$: Each {\em atom} $R(\bm X)$ consists of a
relation symbol $R$ and a set of variables $\bm X\subseteq\bm V$. We assume $\bm V$ to be
the set of all variables in $Q$, i.e., $\bm V = \bigcup_{R(\bm X)\in \atoms(Q)} \bm X$. The
set $\bm F\subseteq \bm V$ is the set of {\em free} variables. A  {\em database instance},
$\calD$, for $Q$ is a mapping that maps each relation symbol $R$ that occurs in an atom
$R(\bm X)\in \atoms(Q)$ to a {\em relation instance}, $R^{\calD}$, which is a finite subset
of $\dom^{\bm X}$. When $\calD$ is clear from the context, we drop the superscript and write
$R$ instead of $R^{\calD}$ to refer to the relation instance, as well as the relation
symbol. The {\em answer} to $Q$ on $\calD$, denoted by $Q(\calD)$, is the relation instance
over $\bm F$ defined as:
\begin{multline*}
    Q(\calD) \defeq \set{\bm f \in \dom^{\bm F} \mid  \exists \bm v \in \dom^{\bm V}
    \text{ where }\bm v_{\bm F} = \bm f, \text{ and } \bm v_{\bm X} \in R^\calD \text{ for all } R(\bm X)\in \atoms(Q)}.
\end{multline*}
A CQ $Q$ is called {\em Boolean} if $\bm F = \emptyset$ and {\em full} if $\bm F = \bm V$.
The {\em size} of a database instance $\calD$, denoted by $N\defeq \norm{\calD}$, is the
total number of tuples in all relations in $\calD$.
Equations~\eqref{eq:4cycle:full} and~\eqref{eq:4cycle} in the introduction are two
examples of CQs: $\cyclefull$ is a full CQ (with $\bm F = \{X,Y,Z,W\}$) and $\cycle$ is a
non-full CQ (with $\bm F = \{X,Y\}$ and existentially quantified variables $Z, W$).

\subsection{Statistics}
Given a relation $R(\bm Z)$ and two disjoint sets of variables $\bm X, \bm Y \subseteq \bm
Z$, the {\em degree of $\bm x$ w.r.t.~$\bm Y$ in $R(\bm Z)$}, for a tuple $\bm x \in
\dom^{\bm X}$, and the {\em degree of $\bm X$ w.r.t.~$\bm Y$ in $R(\bm Z)$} are defined as
follows: ($\pi$ is the projection, and $\sigma$ is the selection operator)
\begin{align*}
    \deg_{R}(\bm Y|\bm X=\bm x) &\quad\defeq\quad
    |\pi_{\bm Y} \sigma_{\bm X = \bm x} R|,\\
    \deg_{R}(\bm Y|\bm X) &\quad\defeq\quad \max_{\bm x \in \dom^{\bm X}} \deg_{R}(\bm Y|\bm X=\bm x).
\end{align*}
A {\em degree constraint} is an expression of the form $\deg_{R}(\bm Y|\bm X) \leq
N_{\bm Y|\bm X}$, where $N_{\bm Y|\bm X}$ is an upper bound on the degree, and $R(\bm Z)$ is
called the {\em guard} of the constraint. Note that $|R(\bm Z)| = \deg_{R}(\bm
Z|\emptyset)$, so a cardinality constraint is a special case of a degree constraint.
Similarly, a functional dependency (FD) $\bm X \to \bm Y$ in $R(\bm Z)$ corresponds to the
degree constraint $\deg_{R}(\bm Y|\bm X) \leq 1$. A {\em set of statistics} $\calS$
is a set of degree constraints; we write $N_{\bm Y|\bm X} \in \calS$ to mean that $\calS$
contains the constraint $\deg_{R}(\bm Y|\bm X) \leq N_{\bm Y|\bm X}$. A database
instance $\calD$ {\em satisfies} $\calS$, written $\calD \models \calS$, if it satisfies all
constraints in $\calS$. We abbreviate $N_{\bm Y|\emptyset}$ as $N_{\bm Y}$. A set of
statistics $\calS$ is called a set of {\em identical cardinality constraints} if it contains
only constraints of the form $N_{\bm Y} = N$, where $N = \norm{\calD}$. Later in
Section~\ref{sec:extensions:lpnorm},
we consider a more general class of statistics.

\subsection{Entropy, polymatroids, and Shannon inequalities}
Given a probability distribution $p$ over $\dom^{\bm V}$, the {\em entropy} of a given set
of variables $\bm X \subseteq \bm V$ is defined as:
\begin{align*}
    h(\bm X) \quad\defeq\quad -\sum_{\bm x \in \dom^{\bm X}} p_{\bm X}(\bm x) \log p_{\bm X}(\bm x)
\end{align*}
where $p_{\bm X}$ is the {\em marginal distribution} of $p$ onto $\bm X$. A function $h:
2^{\bm V} \to \R_+$ is called {\em entropic} if there exists a probability distribution $p$
over $\dom^{\bm V}$ satisfying the above for all $\bm X\subseteq \bm V$. A function $h:
2^{\bm V} \to \R_+$ is called a {\em polymatroid} if it satisfies the following
inequalities, which are known as the {\em basic Shannon inequalities}:
\begin{align}
    h(\emptyset) &= 0 \label{eq:h:emptyset}\\
    h(\bm X) &\leq h(\bm X \cup\bm Y),  &\text{for all } \bm X, \bm Y \subseteq \bm V\label{eq:h:monotonicity}\\
    h(\bm X) + h(\bm Y) &\geq h(\bm X \cup \bm Y) + h(\bm X \cap \bm Y), &\text{for all } \bm X, \bm Y \subseteq \bm V\label{eq:h:submodularity}
\end{align}
Eq.~\eqref{eq:h:monotonicity} is called {\em monotonicity}, and Eq.~\eqref{eq:h:submodularity} is called {\em submodularity}.
Let $\Gamma_n$ denote the set of all polymatroids over $n$ variables, and
$\Gamma^*_n$ denote the set of all entropic functions over $n$ variables.

A {\em Shannon inequality} is any linear inequality that can be derived from the basic
Shannon inequalities, i.e.~ by taking a non-negative linear combination of them. Every
entropic function must satisfy all Shannon inequalities, hence is a polymatroid. However,
not every polymatroid is entropic since entropic functions additionally must satisfy {\em
non-Shannon} inequalities~\cite{DBLP:journals/tit/ZhangY97,zhang1998characterization}. Given
two sets of variables $\bm X, \bm Y$, we abbreviate $\bm X \cup \bm Y$ as $\bm X\bm Y$.
Given a set function $h: 2^{\bm V} \to \R_+$ and two sets of variables $\bm X, \bm Y
\subseteq \bm V$, we define $h(\bm Y|\bm X)\defeq h(\bm X\bm Y) -h(\bm X)$. If $h$ is
entropic then $h(\bm Y|\bm X)$ is the {\em conditional entropy} of $\bm Y$ given $\bm X$. We
call $h(\bm Y|\bm X)$ {\em unconditional} iff $\bm X =\emptyset$, in which case we write
$h(\bm Y)$ instead of $h(\bm Y|\emptyset)$, for brevity. Note that
Inequality~\eqref{eq:h:monotonicity} can be written as $h(\bm Y|\bm X)\geq 0$, whereas
Inequality~\eqref{eq:h:submodularity} can be written as $h(\bm Y|\bm X\cap \bm Y) \geq h(\bm
Y|\bm X)$. By choosing $\bm A\defeq \bm X \cap \bm Y$,\; $\bm B \defeq \bm Y \setminus \bm
X$, and $\bm C \defeq \bm X \setminus \bm Y$, the following is an equivalent form of
Inequality~\eqref{eq:h:submodularity}:
\begin{align}
    h(\bm B|\bm A) \quad\geq\quad h(\bm B|\bm A \bm C), \quad &\text{for all $\bm A, \bm B, \bm C \subseteq \bm V$}
    \label{eq:h:submodularity:equiv}
\end{align}
If $h$ is entropic, then the above intuitively says that the conditional entropy $h(\bm
B|\bm A)$ cannot increase when we condition on {\em extra} variables $\bm C$. In other
words, since entropy measures the uncertainty of a set of variables, knowing more cannot
increase the uncertainty. Let $\calS$ be a set of statistics for a database of size $N$. A
set function $h:2^{\bm V}\to \R_+$ is said to {\em satisfy} $\calS$, denoted $h \models
\calS$, if it satisfies:
\begin{align}
    h(\bm Y|\bm X) \quad\leq\quad \log_N N_{\bm Y|\bm X},& \quad\quad \forall N_{\bm Y|\bm X} \in \calS
    \label{eq:h:satisfies:dc}
\end{align}
Note that in the above, we take the {\em log base-$N$} where $N$ is the input database size.
We use $h \models \calS, \Gamma_n$ to denote that $h$ is a polymatroid that satisfies $\calS$.

\subsection{Acyclic CQs and Tree Decompositions}
A CQ $Q$ (Eq.~\eqref{eq:cq}) is called {\em acyclic}
if we can construct a tree whose nodes are the sets $\bm X$ for $R(\bm X)\in \atoms(Q)$
such that each variable in $\bm V$ appears in a connected subtree.
For example, the following two queries are acyclic:
\begin{align}
    Q_\square^{\calT_1}(X, Y) &\quad\cd\quad A_{11}(X, Y, Z) \wedge A_{12}(Z, W, X)\\
    Q_\square^{\calT_2}(X, Y) &\quad\cd\quad A_{21}(Y, Z, W) \wedge A_{22}(W, X, Y)
\end{align}
An acyclic CQ $Q$ is called {\em free-connex} if
it remains acyclic after adding an extra atom over the free variables $\bm F$.
The above two queries $Q_\square^{\calT_1}$ and $Q_\square^{\calT_2}$ are both free-connex.
Any free-connex CQ $Q$ is known to be evaluatable is time $O(N + \OUT)$ using the Yannakakis
algorithm~\cite{DBLP:conf/vldb/Yannakakis81,DBLP:conf/csl/BaganDG07}, where $\OUT$ is the size of the output $Q(\calD)$.
A {\em tree decomposition (TD)} is a concept used to decompose any CQ (not necessarily acyclic)
into ``subqueries'' whose outputs form an acyclic query.
For our purposes, a TD $\calT$ of a CQ $Q$
is specified by a set, $\bags(\calT)\subseteq 2^{\bm V}$, of variable sets, called the {\em bags}
of the TD, that satisfy the following:
\begin{enumerate}
    \item The bags form an acyclic query $Q(\bm F) \cd$ $\bigwedge_{\bm B \in\bags(\calT)} Q_{\bm B}(\bm B)$, where $Q_{\bm B}$ is a new relation symbol associated with $\bm B$.
    \item Each atom $R(\bm X)\in\atoms(Q)$ is contained in some bag, i.e., satisfies $\bm X \subseteq \bm B$ for some $\bm B \in \bags(\calT)$.
\end{enumerate}
The TD $\calT$ is called {\em free-connex} if
the acyclic query $Q(\bm F) \cd$ $\bigwedge_{\bm B \in\bags(\calT)} Q_{\bm B}(\bm B)$ is free-connex.
If the query is Boolean or full, then all its TDs are free-connex.
For example, consider the 4-cycle query $\cycle$ from Eq.~\eqref{eq:4cycle}:
This query is depicted in Figure~\ref{fig:4cycle} along with two TDs
$\calT_1$ and $\calT_2$
with $\bags(\calT_1) \defeq\{\{X, Y, Z\}, \{Z, W, X\}\}$
and $\bags(\calT_2) \defeq\{\{Y, Z, W\}, \{W, X, Y\}\}$.
These two TDs can be used to transform $\cycle$
into the two acyclic queries $Q_\square^{\calT_1}$ and $Q_\square^{\calT_2}$ defined above.
These two are the only non-trivial TDs of $Q_\square$, besides the trivial one
that puts all variables in one bag.
Both TDs are free-connex.
Given a CQ $Q$, we use $\td(Q)$ to denote the set of all free-connex TDs of $Q$.

\begin{figure}[t]
    \centering
    \begin{tikzpicture}[scale = .4, every node/.style={scale=0.9}]
        \node[] at (0,0) (X) {$X$};
        \node[] at (8,0) (Z) {$Z$};
        \node[] at (4,2) (Y) {$Y$};
        \node[] at (4,-2) (W) {$W$};
        \draw[rotate around={+25:(2, +1)}, ] (2, +1) ellipse (3.2 and .75)
        node[shift={(-.45cm,+.45cm)},blue]{$R$};
        \draw[rotate around={-25:(6, +1)}, ] (6, +1) ellipse (3.2 and .75)
        node[shift={(+.45cm,+.45cm)},blue]{$S$};
        \draw[rotate around={+25:(6, -1)}, ] (6, -1) ellipse (3.2 and .75)
        node[shift={(+.45cm,-.45cm)},blue]{$T$};
        \draw[rotate around={-25:(2, -1)}, ] (2, -1) ellipse (3.2 and .75)
        node[shift={(-.45cm,-.45cm)},blue]{$U$};

        \node[draw,ellipse, right = 3 of Y] (xyz) {$X, Y, Z$};
        \node[draw,ellipse, below = 0.7 of xyz] (zwx) {$Z, W, X$};
        \draw (xyz) -- (zwx);
        \node[below = .1 of zwx] {\color{blue}$\calT_1$};

        \coordinate (TD1) at ($(xyz)!0.5!(zwx)$);
        \node[draw,ellipse, right = 2 of TD1] (wxy) {$W, X, Y$};
        \node[draw,ellipse, right = 0.7 of wxy] (yzw) {$Y, Z, W$};
        \draw (wxy) -- (yzw);
        \coordinate (TD2) at ($(wxy)!0.5!(yzw)$);
        \node[below = 1.1 of TD2] {\color{blue}$\calT_2$};
    \end{tikzpicture}
\caption{Query $Q_\square$ from Eq.~\eqref{eq:4cycle} (or $\cyclefull$ from Eq.~\eqref{eq:4cycle:full}), along with the two free-connex tree decompositions.}
\label{fig:4cycle}
\end{figure}

\section{Cost Estimation of Static Query Plans}
\label{sec:static:query:plans}

\subsection{Static Query Plans}
\label{subsec:static:query:plans}

A {\em static} query plan evaluates a CQ $Q$ using a single tree decomposition (TD),
subsuming the classical notion of a join-project plan based on a join tree.
The cost of a plan is the {\em worst-case cardinality of the largest intermediate relation}
produced during execution; estimating this cost reduces to bounding the worst-case output
size of a CQ (see Section~\ref{subsec:output-size-bound}).
Optimizing over all TDs leads naturally to the {\em fractional hypertree width}
of $Q$ under $\calS$~\cite{DBLP:journals/talg/GroheM14}, which we generalize here
to arbitrary statistics and non-Boolean CQs.

Consider the conjunctive query~\eqref{eq:cq}, over a database instance $\calD$
satisfying statistics $\calS$. Given a tree decomposition $\calT \in
\td(Q)$, the query plan that $\calT$ represents can be expressed via the following two
rules:
\begin{align}
    \bigwedge_{\bm B\in\bags(\calT)} Q_{\bm B}(\bm B)
        &\quad\cd\quad \bigwedge_{R(\bm X)\in\atoms(Q)} R(\bm X) \label{eq:static:plans:rule1}\\
    Q(\bm F) &\quad\cd\quad \bigwedge_{\bm B\in\bags(\calT)} Q_{\bm B}(\bm B) \label{eq:static:plans:rule2}
\end{align}
Rule~\eqref{eq:static:plans:rule1} is somewhat unconventional, with a conjunction in the head;
we formulate it this way for reasons that will be clear later. Its semantics is as expected:
relations $Q_{\bm B}$ are answers to the query if, for every tuple $\bm t$ satisfying its
body, the conjunction $\bigwedge_{\bm B\in\bags(\calT)}$ $Q_{\bm B}(\bm t_{\bm B})$ holds. In
particular, it asks us to compute one intermediate relation $Q_{\bm B}$ for each bag $\bm B
\in \bags(\calT)$. We will assume that the $Q_{\bm B}$ are minimal answers; then, we compute
$Q(\bm F)$ using rule~\eqref{eq:static:plans:rule2}. Since $\calT$ is free-connex,
rule~\eqref{eq:static:plans:rule2} defines an acyclic free-connex query over the intermediate
relations $\{Q_{\bm B}\}_{\bm B\in\bags(\calT)}$, and can therefore be evaluated in time
$O(\max_{\bm B\in\bags(\calT)} |Q_{\bm B}| + |Q(\bm F)|)$ using the Yannakakis
algorithm~\cite{DBLP:conf/vldb/Yannakakis81,DBLP:conf/csl/BaganDG07}.

Equivalently, rule~\eqref{eq:static:plans:rule1} can be evaluated by computing
the intermediate relations $Q_{\bm B}$ using the following CQs,
one for each bag $\bm B \in \bags(\calT)$:
\begin{align}
    Q_{\bm B}(\bm B) &\quad\cd\quad \bigwedge_{R(\bm X)\in\atoms(Q)} R(\bm X),
    && \bm B \in \bags(\calT)
    \label{eq:single-td:rule1:individual}
\end{align}

Ideally, we would like the cost of the query plan $\calT$ to be the worst-case cardinality of the
largest intermediate relation, over all database instances satisfying $\calS$:
\begin{align}
    \sup_{\calD \models \calS}
    \; \max_{\bm B \in \bags(\calT)}\; |Q_{\bm B}(\calD)|
    \quad=\quad \max_{\bm B \in \bags(\calT)}\; \sup_{\calD \models \calS}\;  |Q_{\bm B}(\calD)|
    \label{eq:cost:single-td}
\end{align}
where $Q_{\bm B}$ is the output of rule~\eqref{eq:single-td:rule1:individual}. Consequently,
the first thing we need to do is to estimate $\sup_{\calD \models \calS}  |Q_{\bm
B}(\calD)|$ for each bag $\bm B$. This is the problem of bounding the worst-case output size
of a CQ $Q$ under statistics $\calS$, which we discuss next.

\subsection{Output Cardinality Bound of a Conjunctive Query}
\label{subsec:output-size-bound}

Given a CQ $Q$ of the form~\eqref{eq:cq} and statistics $\calS$, this section develops
bounds on the worst-case output cardinality of $Q$:
\begin{align}
    \sup_{\calD \models \calS} |Q(\calD)|
    \label{eq:worst-case-output-size}
\end{align}
The bounds are based on a long line of research using entropic arguments
in extremal combinatorics~\cite{MR599482, MR859293, MR1639767, friedgut-kahn-1998,
DBLP:journals/siamcomp/AtseriasGM13, DBLP:journals/jacm/GottlobLVV12,
DBLP:conf/pods/KhamisNS16, DBLP:conf/pods/Khamis0S17, theoretics:13722, pandaexpress}.
We introduce the {\em entropic bound} and its tractable relaxation, the {\em polymatroid
bound}, as the main tools.
To demonstrate the idea, consider the 4-cycle query $\cyclefull$ from Eq.~\eqref{eq:4cycle:full}.
Suppose we were given the following {\em statistics} about the input database instance:
\begin{itemize}[leftmargin=*]
    \item The size of each relation is at most $N$ for some number $N$.
    \item The relation $U(W, X)$ has a functional dependency (FD) $W\to X$, i.e.~$\deg_U(X|W) \leq 1$.
    \item The relation $U(W, X)$ has a degree constraint $\deg_{U}(W|X) \leq C$
    for some number $C$.
\end{itemize}
In particular,
the given set of statistics $\cyclefullstats$ is as follows:
\begin{align}
    \cyclefullstats \quad\defeq\quad \{N_{XY} = N_{YZ} = N_{ZW} = N_{WX} = N,\quad N_{X|W} = 1,\quad N_{W|X} = C\}
    \label{eq:4cycle:full:stats}
\end{align}
Consider a database instance $\calD\models\cyclefullstats$. Let's say we want to estimate
$\sup_{\calD\models\cyclefullstats} |\cyclefull(\calD)|,$ which is an instance
of~\eqref{eq:worst-case-output-size} for the query $\cyclefull$ and statistics
$\cyclefullstats$. Figure~\ref{fig:4cycle:data} depicts such an instance along with the
corresponding output $\cyclefull(\calD)$.

The entropy argument goes this like: take a {\em uniform} probability distribution
over the output tuples in $\cyclefull(\calD)$, and use it to define an {\em entropic
function} $h:2^{\bm V}\to \R_+$ where $\bm V \defeq \{X, Y, Z, W\}$ in this example. By
uniformity, we have $h(XYZW) = \log |\cyclefull(\calD)|$. Moreover, since the projection of
$\cyclefull(\calD)$ onto $(X, Y)$ is contained in $R$, we have $h(XY) \leq \log |R| \leq
\log N$. Similarly, we have $h(YZ), h(ZW), h(WX) \leq \log N$. In addition, since $U(W, X)$
has a functional dependency $W\to X$, this means that the {\em conditional entropy} of $X$
given $W$ is zero, i.e., $h(X|W) = 0$. Finally, since $\deg_U(W|X) \leq C$, we have $h(W|X)
\leq \log C$. For example, if every $X$-value in $U$ has at most 2 matching $W$-values as in
Figure~\ref{fig:4cycle:data}, then we need one bit of information to represent $W$ for a
given $X$, i.e. $h(W|X) = \log 2 = 1$. Putting everything together, the entropic function
$h$ satisfies the statistics in $\cyclefullstats$ {\em if} we scale it down by a factor of
$\log N$, in order to match Eq.~\eqref{eq:h:satisfies:dc}. In particular, if we define $\ov
h \defeq \frac{1}{\log N}\cdot h$, then $\ov h\models\cyclefullstats$ is entropic\footnote{In general, $\ov h$ may not be entropic, but it is always a limit of entropic functions.}, and
satisfies $\ov h(XYZW) = \log_N |\cyclefull(\calD)|$.

\begin{figure*}
    \begin{minipage}[t]{0.18\textwidth}
        \centering
        $R\quad\quad$~\\\vspace{-0.25cm}
        \begin{tabular}[t]{|c|c|c}
            \cline{1-2}
            \rowcolor{lightgray}
            $X$ & $Y$ \\\cline{1-2}
            $1$ & $p$ & \color{red}{1/3}\\
            $1$ & $q$ & \color{red}{2/3}\\
            $2$ & $p$ & \color{red}{0}\\\cline{1-2}
        \end{tabular}
    \end{minipage}
    \begin{minipage}[t]{0.18\textwidth}
        \centering
        $S\quad\quad$~\\\vspace{-0.25cm}
        \begin{tabular}[t]{|c|c|c}
            \cline{1-2}
            \rowcolor{lightgray}
            $Y$ & $Z$ \\\cline{1-2}
            $p$ & $3$ &\color{red}{1/3}\\
            $q$ & $4$ &\color{red}{0}\\
            $q$ & $5$ &\color{red}{2/3}\\\cline{1-2}
        \end{tabular}
    \end{minipage}
    \begin{minipage}[t]{0.18\textwidth}
        \centering
        $T\quad\quad$~\\\vspace{-0.25cm}
        \begin{tabular}[t]{|c|c|c}
            \cline{1-2}
            \rowcolor{lightgray}
            $Z$ & $W$ \\\cline{1-2}
            $3$ & $i$ &\color{red}{1/3}\\
            $5$ & $i$ &\color{red}{1/3}\\
            $5$ & $j$ &\color{red}{1/3}\\\cline{1-2}
        \end{tabular}
    \end{minipage}
    \begin{minipage}[t]{0.18\textwidth}
        \centering
        $U\quad\quad$~\\\vspace{-0.25cm}
        \begin{tabular}[t]{|c|c|c}
            \cline{1-2}
            \rowcolor{lightgray}
            $W$ & $X$ \\\cline{1-2}
            $i$ & $1$ &\color{red}{2/3}\\
            $j$ & $1$ &\color{red}{1/3}\\
            $k$ & $2$ &\color{red}{0}\\\cline{1-2}
        \end{tabular}
    \end{minipage}
    \begin{minipage}[t]{0.25\textwidth}
        \centering
        $\cyclefull\quad\quad$~\\\vspace{-0.25cm}
        \begin{tabular}[t]{|c|c|c|c|c}
            \cline{1-4}
            \rowcolor{lightgray}
            $X$ & $Y$ & $Z$ & $W$ \\\cline{1-4}
            $1$ & $p$ & $3$ & $i$ &\color{red}{1/3}\\
            $1$ & $q$ & $5$ & $i$ &\color{red}{1/3}\\
            $1$ & $q$ & $5$ & $j$ &\color{red}{1/3}\\ \cline{1-4}
        \end{tabular}
    \end{minipage}
    \caption{An input database instance for the 4-cycle query $\cyclefull$ from Eq.~\eqref{eq:4cycle:full} along with the corresponding output.
    Each output tuple is annotated with its probability (the {\color{red} red} numbers to the right) under a {\em uniform} distribution.
    Similarly, each input tuple is annotated with its {\em marginal}
    probability.}
    \label{fig:4cycle:data}
\end{figure*}

Since $\ov h$ is entropic and $\ov h \models \cyclefullstats$, we can thus bound:
\begin{align}
    \log_N |\cyclefull(\calD)| \quad\leq\quad \sup_{h \models \cyclefullstats, \Gamma^*_4} h(XYZW)
\end{align}

The above discussion generalizes to {\em any} CQ $Q$ (not necessarily full) and {\em any} statistics $\calS$. For any
database instance $\calD\models\calS$, there exists an entropic function $h\models\calS$
such that $h(\bm F) = \log_N |Q(\calD)|$, where $\bm F\subseteq\bm V$ are the free variables of $Q$. Hence, the following is indeed an upper bound (on
$\log_N$-scale) on the maximum output size of $Q$ over all database instances
$\calD\models\calS$. We refer to it as the {\em entropic bound} of $Q$ under $\calS$:


\begin{thm}[\cite{DBLP:conf/pods/Khamis0S17}]
Let $Q$ be a CQ (Eq.~\eqref{eq:cq}), with given statistics $\calS$ from an input database instance. Then, we have
\begin{align}
    \log_N {\color{blue}\underbrace{\color{black}\sup_{\calD\models\calS} |Q(\calD)|}_{\substack{\text{worst-case}\\\text{output size}}}}
    \leq
    {\color{blue}\underbrace{\color{black}\sup_{h \models \calS, \Gamma^*_n} h(\bm F)}_{\text{Entropic-bound$(Q, \calS)$}}}
    \leq
    {\color{blue}\underbrace{\color{black}\max_{h \models \calS, \Gamma_n} h(\bm F)}_{\text{Polymatroid-bound$(Q, \calS)$}}}
    \label{eq:entropic:polymatroid:bound}
\end{align}
\label{thm:cq:bound}
\end{thm}

The entropic bound is asymptotically
tight~\cite{DBLP:journals/tit/ChanY02,DBLP:conf/pods/Khamis0S17}. However, it is not known
to be computable since characterizing the entropic cone is a major long-standing open
problem in information theory~\cite{kolaitis_et_al:DagRep.12.7.180}. To overcome the
computability issue, we relax the bound to replace the set of entropic functions with the
set of polymatroids, which is a {\em superset} of entropic functions. Recall that {\em
polymatroids} are only required to satisfy Shannon inequalities
(Eq.~\eqref{eq:h:emptyset}--\eqref{eq:h:submodularity}), whereas entropic functions must
additionally satisfy {\em non-Shannon
inequalities}~\cite{DBLP:journals/tit/ZhangY97,zhang1998characterization}. Because
polymatroids are a superset, the polymatroid bound is an upper bound on the entropic bound.
It is {\em not} tight in general~\cite{DBLP:conf/pods/Khamis0S17}. However, it is easily
computable as a linear program (LP), since Shannon inequalities as well as
Eq.~\eqref{eq:h:satisfies:dc} are all linear constraints. In particular, {\em going from the
entropic to the polymatroid bound, we trade tightness for computability}. The polymatroid
bound builds on top of prior bounds introduced in the
literature~\cite{DBLP:journals/siamcomp/AtseriasGM13,DBLP:journals/jacm/GottlobLVV12,DBLP:conf/pods/GottlobLV09}:
\begin{itemize}[leftmargin=*]
    \item If all degree constraints are cardinality constraints, then the polymatroid bound collapses
back to the {\em AGM bound}~\cite{DBLP:journals/siamcomp/AtseriasGM13}.
In this case, it also coincides with the entropic bound, thus making it tight asymptotically.
    \item If degree constraints contain only cardinality constraints and FDs, then the polymatroid bound coincides with the bound proposed by Gottlob, Lee, Valiant and Valiant, also known as the {\em GLVV bound}~\cite{DBLP:journals/jacm/GottlobLVV12,DBLP:conf/pods/GottlobLV09}.
    In this case, the polymatroid bound is already {\em not} tight~\cite{DBLP:conf/pods/Khamis0S17}.
\end{itemize}
In our example, the polymatroid bound implies the following bound on the output size of $\cyclefull$:
\begin{align}
    |\cyclefull(\calD)| \quad\leq\quad N^{3/2}\cdot\sqrt{C}, \quad\quad\quad \forall \calD\models\cyclefullstats
\end{align}
This is because any polymatroid $h$ must satisfy the following Shannon inequality, where
each term on the RHS is upper bounded by some constraint in $h \models\cyclefullstats$:
\begin{align}
    h(XYZW) \leq \frac{1}{2}\biggl[{\color{blue}\underbrace{\color{black}h(XY)}_{\leq 1}}+
    {\color{blue}\underbrace{\color{black}h(YZ)}_{\leq 1}}+{\color{blue}\underbrace{\color{black}h(ZW)}_{\leq 1}} + {\color{blue}\underbrace{\color{black}h(X|W)}_{=0}} + {\color{blue}\underbrace{\color{black}h(W|X)}_{\leq \log_N C}}\biggr]
\end{align}
We leave it as an exercise to verify that the above inequality is indeed a Shannon
inequality, i.e.~follows from Eq.~\eqref{eq:h:monotonicity}--\eqref{eq:h:submodularity}. We
will show later how to solve (more general forms of) the LP corresponding to the polymatroid
bound, and how to come up with these Shannon inequalities.

\subsection{Costing Static Query Plans}
\label{subsec:costing:single-td}

Getting back to the problem of estimating the cost of a static query plan $\calT$ under
statistics $\calS$: we do not know how to compute the quantity in~\eqref{eq:cost:single-td}
efficiently, and thus we use the polymatroid bound~\eqref{eq:entropic:polymatroid:bound} to
estimate it instead, and define the cost $\cost(\calT, \calS)$ of the query plan $\calT$
under statistics $\calS$ as (in $\log_N$ scale):
\begin{align}
    \cost(\calT, \calS)
    \quad=\quad
    \max_{\bm B \in \bags(\calT)}\; \max_{h \models \calS, \Gamma_n} h(\bm B)
    \label{eq:cost:single-td:polymatroid}
\end{align}
Note that the $\sup$ was replaced by $\max$ because the optimization problem is
now a linear program (LP) with a compact representation.

Since we want to find the best query plan, we seek to minimize $\cost(\calT,
\calS)$ defined in~\eqref{eq:cost:single-td:polymatroid} over all TDs $\calT \in \td(Q)$,
giving rise to the following measure of complexity, which we call the {\em fractional
hypertree width} of $Q$ under $\calS$:
\begin{align}
    \fhtw(Q, \calS) \defeq
    {\color{blue}\underbrace{\color{black}\min_{\calT \in \td(Q)}}_{\substack{\text{best single-TD}\\\text{query plan}}}} \quad
    {\color{blue}\underbrace{\color{black}\max_{\bm B \in \bags(\calT)}}_{\substack{\text{worst subquery}\\\text{in the plan}}}} \quad
    {\color{blue}\underbrace{\color{black}\max_{h \models \calS, \Gamma_n} h(\bm B)}_{\substack{\text{Polymatroid bound}\\\text{of the subquery}}}}
    \label{eq:fhtw}
\end{align}
While not immediately obvious, this definition generalizes the classical {\em fractional
hypertree width} of Grohe and Marx~\cite{DBLP:journals/talg/GroheM14} in two ways:
(1) it applies to arbitrary statistics $\calS$, whereas the original definition only
considers identical cardinality constraints; and
(2) it applies to arbitrary CQs $Q$, whereas the original definition is restricted to
Boolean CQs.
We will show later that $\panda$ can evaluate any CQ $Q$ in time
$O(N^{\fhtw(Q, \calS)} \cdot \log N + \OUT)$.

For example, consider the query $\cycle$ from Eq.~\eqref{eq:4cycle} and the tree decomposition $\calT_1$
from Figure~\ref{fig:4cycle}, with bags $\{X,Y,Z\}$ and $\{Z,W,X\}$.
Suppose that the given statistics $\calS_\square$ only contain identical cardinality constraints:
\begin{align}
    \cyclestats \quad\defeq\quad \{N_{XY} = N_{YZ} = N_{ZW} = N_{WX} = N\}
    \label{eq:4cycle:stats}
\end{align}
The query plan induced by $\calT_1$ is realized via the following three rules:
\begin{align}
    A(X, Y, Z) &\quad\cd\quad R(X, Y) \wedge S(Y, Z) \wedge T(Z, W) \wedge U(W, X) \label{eq:4cycle:plan:A}\\
    B(Z, W, X) &\quad\cd\quad R(X, Y) \wedge S(Y, Z) \wedge T(Z, W) \wedge U(W, X) \label{eq:4cycle:plan:B}\\
    \cycle(X, Y) &\quad\cd\quad A(X, Y, Z) \wedge B(Z, W, X) \label{eq:4cycle:plan:Q}
\end{align}
Rules~\eqref{eq:4cycle:plan:A} and~\eqref{eq:4cycle:plan:B} are instances of
rule~\eqref{eq:single-td:rule1:individual}, computing one intermediate relation per bag of
$\calT_1$. Rule~\eqref{eq:4cycle:plan:Q} is an instance of~\eqref{eq:static:plans:rule2},
and can be evaluated in time $O(|A| + |B| + |\cycle(\calD)|)$ by the Yannakakis algorithm,
since $\calT_1$ is free-connex.
We would like the cost of this plan to be
\begin{align*}
        \sup_{\calD \models \calS_\square}\, \max(|A(\calD)|, |B(\calD)|)
        =
        \max \{
            \sup_{\calD \models \calS_\square} |A(\calD)|,
            \sup_{\calD \models \calS_\square} |B(\calD)|
        \}
\end{align*}
However, we will use the polymatroid bound to estimate it instead, giving us (in $\log_N$ scale):
\begin{align}
    \cost(\calT_1, \calS_\square)
    &\defeq \max\!\left(\max_{h \models \calS_\square, \Gamma_4} h(XYZ),\; \max_{h \models \calS_\square, \Gamma_4} h(ZWX)\right) \label{eq:cost:4cycle:plan1:bound}
\end{align}

It is not hard to see that the polymatroid bound for each bag is:
\begin{align*}
    \max_{h \models \cyclestats, \Gamma_4} h(XYZ) = 2,
    \qquad
    \max_{h \models \cyclestats, \Gamma_4} h(ZWX) = 2
\end{align*}
so $\cost(\calT_1, \cyclestats) = 2$.
(Note that we are using statistics $\cyclestats$ from Eq.~\eqref{eq:4cycle:stats} instead of $\cyclefullstats$ from Eq.~\eqref{eq:4cycle:full:stats}.)
By symmetry, the same holds for $\calT_2$ with bags $\{Y,Z,W\}$ and $\{W,X,Y\}$, giving
$\cost(\calT_2, \cyclestats) = 2$ as well.
Since $\calT_1$ and $\calT_2$ are the only non-trivial free-connex TDs of $\cycle$
(Figure~\ref{fig:4cycle}), minimizing over all TDs gives $\fhtw(\cycle, \cyclestats) = 2$.
%

\section{Cost Estimation of Adaptive Query Plans}
\label{sec:adaptive:query:plans}

\subsection{Adaptive Query Plans}
\label{subsec:adaptive:query:plans}

An {\em adaptive} query plan evaluates a CQ $Q$ using multiple tree decompositions, where the
input data {\em and} intermediate relations are partitioned across the TDs for load
balancing. This notion of query plans is highly novel, and provides a principled and general
formalization of data partitioning strategies. Costing these plans leads to bounding the
worst-case output size of {\em disjunctive datalog rules} (DDRs), and the cost function leads
naturally to the notion of {\em submodular width} of $Q$ under $\calS$, which we generalize
from the original definition by Marx~\cite{DBLP:journals/jacm/Marx13} to arbitrary
statistics and non-Boolean CQs.

Our motivation for adaptive query plans is as follows. For some queries, using a single TD
might not be sufficient to achieve the best possible runtime. For example, consider $\cycle$
from Eq.~\eqref{eq:4cycle} and $\cyclestats$ from Eq.~\eqref{eq:4cycle:stats}. From the previous section, $\fhtw(Q_\square, \cyclestats) = 2$,
which indicates that there are database instances $\calD\models\cyclestats$ where no matter
which TD we pick (Figure~\ref{fig:4cycle}), there will always be a bag whose size is
$\Omega(N^2)$. Indeed, here is one such instance $\calD\models\cyclestats$, assuming $N$ is
even:~\footnote{We use $[N]$ to denote the set $\{1, 2, \ldots, N\}$.}
\begin{align*}
    R = S = T = U = ([N/2] \times [1]) \cup ([1] \times [N/2])
\end{align*}
%
%
Nevertheless, $Q_\square$ admits a faster $O(N^{3/2})$ algorithm that partitions the input
data and uses a different TD for each part~\cite{DBLP:journals/algorithmica/AlonYZ97}.
Partitioning in this case is done based on some degrees being smaller or larger than some
threshold, e.g., $\deg_{R}(Y|X =x) \leq \sqrt{N}$. We refer to query plans that partition
the data and use different TDs as {\em adaptive} query plans. Partitioning is not limited to
input relations only: We can partition any intermediate relation that is computed along the
way. It is also not limited to comparing degrees to fixed thresholds: We could use more
advanced partitioning strategies.

More generally, to answer the CQ~\eqref{eq:cq}, instead of committing to a single TD as in
rule~\eqref{eq:static:plans:rule1}, an adaptive query plan can be expressed via the following
two rules:
\begin{align}
    \bigvee_{\calT \in \td(Q)} \bigwedge_{\bm B\in\bags(\calT)} Q_{\bm B}(\bm B)
        &\cd \bigwedge_{R(\bm X)\in\atoms(Q)} R(\bm X) \label{eq:multi-td:rule1}\\
    Q(\bm F) &\cd \bigvee_{\calT \in \td(Q)} \bigwedge_{\bm B\in\bags(\calT)} Q_{\bm B}(\bm B) \label{eq:multi-td:rule2}
\end{align}
Rule~\eqref{eq:multi-td:rule1} is a {\em disjunctive} rule, which is even more
unconventional than rule~\eqref{eq:static:plans:rule1}. Its semantics is the natural one:
the relations $Q_{\bm B}$ are feasible solution to the query if, for every tuple $\bm t$
satisfying the body, there exists at least one TD $\calT \in \td(Q)$ for which the
conjunction $\bigwedge_{\bm B \in \bags(\calT)} Q_{\bm B}(\bm t_{\bm B})$ holds. We will
discuss how to evaluate rule~\eqref{eq:multi-td:rule1} later; for now, note that once the
relations $Q_{\bm B}$ are computed, the answer to $Q$ can be obtained by evaluating
rule~\eqref{eq:multi-td:rule2}, which is a union of CQs. Since all TDs
$\calT \in \td(Q)$ are free-connex, rule~\eqref{eq:multi-td:rule2} can be evaluated in time
$O(\max_{\bm B} |Q_{\bm B}| + |Q(\bm F)|)$.
Thus, the dominant cost is still the size of the largest intermediate relation $Q_{\bm B}$,
which is the output of rule~\eqref{eq:multi-td:rule1}.

As an illustration, consider the query $\cycle$ from Eq.~\eqref{eq:4cycle}. Its adaptive
query plan, built on top of the two TDs $\calT_1$ and $\calT_2$ shown in
Figure~\ref{fig:4cycle}, is represented by the following two rules:
\begin{multline}
    \bigl(A_{11}(X,Y,Z) \wedge A_{12}(Z,W,X)\bigr) \vee \bigl(A_{21}(Y,Z,W) \wedge A_{22}(W,X,Y)\bigr) \\
    \cd\quad R(X,Y) \wedge S(Y,Z) \wedge T(Z,W) \wedge U(W,X) \label{eq:4cycle:adaptive:rule1}
\end{multline}
\begin{align}
    \cycle(X,Y) \quad\cd\quad &\bigl(A_{11}(X,Y,Z) \wedge A_{12}(Z,W,X)\bigr) \vee \bigl(A_{21}(Y,Z,W) \wedge A_{22}(W,X,Y)\bigr) \label{eq:4cycle:adaptive:rule2}
\end{align}
where $A_{11}(X,Y,Z)$ and $A_{12}(Z,W,X)$ are the intermediate relations for the bags of
$\calT_1$, and $A_{21}(Y,Z,W)$ and $A_{22}(W,X,Y)$ are those for the bags of $\calT_2$.
As we shall see later, rule~\eqref{eq:4cycle:adaptive:rule1} admits an answer of size
$O(N^{3/2})$, matching the best known algorithm for this query~\cite{DBLP:journals/algorithmica/AlonYZ97}.

Next, we discuss how to evaluate the disjunctive rule~\eqref{eq:multi-td:rule1}, and how to
cost it. To this end, we need the concept of {\em bag selectors}, where a {\em bag selector
$\calB$} is a tuple of subsets of $\bm V$, consisting of exactly one bag from each TD in
$\td(Q)$. We use $\bs(Q)$ to denote the set of all bag selectors of $Q$. Using this notion,
and the fact that $\vee$ distributes over $\wedge$, we can rewrite the head of the
rule~\eqref{eq:multi-td:rule1} as follows:
\begin{align}
    \bigvee_{\calT \in \td(Q)} \bigwedge_{\bm B\in\bags(\calT)} Q_{\bm B}(\bm B)
    \equiv
    \bigwedge_{\calB \in \bs(Q)}
    \bigvee_{\bm B \in \calB}
    Q_{\bm B}(\bm B)
\end{align}
And thus, the rule can be reformulated as:
\begin{align}
    \bigwedge_{\calB \in \bs(Q)}
    \bigvee_{\bm B \in \calB}
    Q_{\bm B}(\bm B)
        &\quad\cd\quad \bigwedge_{R(\bm X)\in\atoms(Q)} R(\bm X) \label{eq:multi-td:rule1:rewritten}
\end{align}
Equivalently, rule~\eqref{eq:multi-td:rule1:rewritten} is a collection of rules of the form
\begin{align}
    \bigvee_{\bm B \in \calB}
    Q_{\bm B}(\bm B)
        &\quad\cd\quad \bigwedge_{R(\bm X)\in\atoms(Q)} R(\bm X),
    && \calB \in \bs(Q)
   \label{eq:ddr:rule}
\end{align}
one for each $\calB \in \bs(Q)$. These rules~\eqref{eq:ddr:rule} are called
{\em disjunctive datalog rules} (DDR).

For example, for the $\cycle$ query, rule~\eqref{eq:4cycle:adaptive:rule1} has two TDs, each with two
bags, so $\bs(Q_\square)$ consists of four bag selectors --- one bag chosen from each TD.
Using distributivity of $\vee$ over $\wedge$, we can rewrite the head of
rule~\eqref{eq:4cycle:adaptive:rule1} as:
\begin{align*}
    &\bigl(A_{11}(X,Y,Z) \wedge A_{12}(Z,W,X)\bigr) \vee \bigl(A_{21}(Y,Z,W) \wedge A_{22}(W,X,Y)\bigr) \\
    \equiv&
    \bigl(A_{11}(X,Y,Z) \vee A_{21}(Y,Z,W)\bigr) \wedge
    \bigl(A_{11}(X,Y,Z) \vee A_{22}(W,X,Y)\bigr) \wedge\\
    &
    \bigl(A_{12}(Z,W,X) \vee A_{21}(Y,Z,W)\bigr) \wedge
    \bigl(A_{12}(Z,W,X) \vee A_{22}(W,X,Y)\bigr)
\end{align*}
And thus, rule~\eqref{eq:4cycle:adaptive:rule1} can be reformulated into $4$
disjunctive datalog rules, one for each bag selector $\calB \in \bs(Q_\square)$:
{\small
\begin{align*}
    A_{11}(X,Y,Z) \vee A_{21}(Y,Z,W)
    &\cd R(X,Y) \wedge S(Y,Z) \wedge T(Z,W) \wedge U(W,X)\\
    A_{11}(X,Y,Z) \vee A_{22}(W,X,Y)
    &\cd R(X,Y) \wedge S(Y,Z) \wedge T(Z,W) \wedge U(W,X)\\
    A_{12}(Z,W,X) \vee A_{21}(Y,Z,W)
    &\cd R(X,Y) \wedge S(Y,Z) \wedge T(Z,W) \wedge U(W,X)\\
    A_{12}(Z,W,X) \vee A_{22}(W,X,Y)
    &\cd R(X,Y) \wedge S(Y,Z) \wedge T(Z,W) \wedge U(W,X)
\end{align*}}

Just like in the static query plan case, ideally we would like to use the worst-case output
size bounds of the above disjunctive datalog rules to cost the adaptive query plan. Unlike
in the case of a conjunctive query, where the output of the query given a database $\calD$
is unique (it's the unique minimal model), a DDR can have multiple minimal models (models
are feasible solutions to the DDR). We thus measure the output size of the
DDR~\eqref{eq:ddr:rule} by the quantity
\begin{align}
    \min_{(Q_{\bm B})_{\bm B\in\calB}}\quad \max_{\bm B \in \calB}\quad |Q_{\bm B}(\calD)|
\end{align}
where the $\min$ is taken over all models of the DDR~\eqref{eq:ddr:rule} given $\calD$.
In particular, we would like to estimate
\begin{align}
    \sup_{\calD\models\calS}\quad
        \max_{\calB \in \bs(Q)} \quad\min_{(Q_{\bm B})_{\bm B\in\calB}}\quad
            \max_{\bm B \in \calB}\quad
            |Q_{\bm B}(\calD)|
    \quad=\quad
    \max_{\calB \in \bs(Q)}\quad
    \sup_{\calD\models\calS}\quad
    \min_{(Q_{\bm B})_{\bm B\in\calB}}\quad \max_{\bm B \in \calB}\quad |Q_{\bm B}(\calD)|
    \label{eq:multi-td:bound:estimate}
\end{align}
This is the subject of the next section.

\subsection{Output Cardinality Bound of a Disjunctive Datalog Rule}
\label{subsec:ddr:bound}

The problem we want to solve in this subsection is to estimate, for a given disjunctive
datalog rule (DDR) of the form~\eqref{eq:ddr:rule} and statistics $\calS$, the following
quantity:
\begin{align}
    \sup_{\calD \models \calS}\quad
    \min_{(Q_{\bm B})_{\bm B\in\calB}}\quad
    \max_{\bm B \in \calB}\quad |Q_{\bm B}(\calD)|
    \label{eq:ddr:bound}
\end{align}
To illustrate the idea, consider the following DDR from the $\cycle$ example:
\begin{align}
    A_{11}(X,Y,Z) \vee A_{21}(Y,Z,W) \cd R(X,Y) \wedge S(Y,Z) \wedge T(Z,W) \wedge U(W,X)
    \label{eq:4cycle:ddr:example}
\end{align}
Here, we want to estimate the quantity
\begin{align}
    \sup_{\calD \models \calS}\; \min_{A_{11}, A_{21}}\; \max\bigl(|A_{11}(\calD)|,\, |A_{21}(\calD)|\bigr)
    \label{eq:4cycle:ddr:example:bound}
\end{align}
where the $\min$ is taken over all relations $A_{11}$ and $A_{21}$ that are models of
the DDR~\eqref{eq:4cycle:ddr:example} given $\calD$.

We prove an information-theoretic upper bound for this quantity by constructing a specific
solution $(\bar{A}_{11}, \bar{A}_{21})$ and a probability distribution over $(X,Y,Z,W)$, as follows.
Initialize tables $\bar{A}_{11}$, $\bar{A}_{21}$, and an auxiliary relation $\bar{A}(X,Y,Z,W)$ to be empty.
For every tuple $(x,y,z,w)$ satisfying the body $R(X,Y) \wedge S(Y,Z) \wedge T(Z,W) \wedge
U(W,X)$, do the following:
\begin{itemize}
    \item If $(x,y,z) \in \bar{A}_{11}$ or $(y,z,w) \in \bar{A}_{21}$, do nothing.
    \item Otherwise, insert $(x,y,z)$ into $\bar{A}_{11}$, $(y,z,w)$ into $\bar{A}_{21}$, and
    $(x,y,z,w)$ into $\bar{A}$.
\end{itemize}
Now take a uniform distribution over the tuples in $\bar{A}$, and let $h$ be its entropy
function.\footnote{Notice that the runtime of this algorithm is larger than the output size of the DDR because it produces the output to the {\em full} CQ $\cyclefull$; this algorithm is only meant as proof technique to construct the entropic vector $h$} By uniformity, we have:
\begin{align*}
    h(XYZ) = h(YZW) = h(XYZW) = \log_2 |\bar{A}_{11}| = \log_2 |\bar{A}_{21}| = \log_2 |\bar{A}|.
\end{align*}
Furthermore, the support of the marginal distribution on $(X,Y)$ is a subset of $R$, so
$h(XY) \leq \log_2 |R| = \log_2 N$. Similarly, $h(YZ) \leq \log_2 N$, $h(ZW) \leq \log_2
N$, and $h(WX) \leq \log_2 N$.
Thus, $\bar h := \frac{1}{\log N} h$ is a polymatroid such that $\bar h \models \calS_\square$
and $\bar h(XYZ) = \bar h(YZW) = \log_N |\bar{A}_{11}| = \log_N |\bar{A}_{21}|$.

We can therefore conclude with the following bound for the DDR~\eqref{eq:4cycle:ddr:example}:
\begin{align*}
    &\log_N \sup_{\calD \models \calS_\square}\; \min_{A_{11}, A_{21}}\; \max\bigl(|A_{11}(\calD)|,\, |A_{21}(\calD)|\bigr)\\
    &=
    \sup_{\calD \models \calS_\square}\; \min_{A_{11}, A_{21}}\; \max\bigl(\log_N |A_{11}(\calD)|,\, \log_N |A_{21}(\calD)|\bigr)\\
    &\leq
    \sup_{\calD \models \calS_\square}\; \max\bigl(\log_N |\bar{A}_{11}(\calD)|,\, \log_N |\bar{A}_{21}(\calD)|\bigr) \\
    &=
    \sup_{\calD \models \calS_\square}\; \log_N |\bar{A}_{11}(\calD)| \\
    &=
    \sup_{\calD \models \calS_\square}\; \min \bigl(\bar h(XYZ), \bar h(YZW)\bigr) \\
    &\leq
    \sup_{h \models \calS_\square, \Gamma_4}\; \min \bigl(h(XYZ), h(YZW)\bigr).
\end{align*}

The same proof technique generalizes to any DDR, yielding the following bound.

\begin{theorem}[Polymatroid bound of a DDR~\cite{DBLP:conf/pods/Khamis0S17}]
\label{thm:ddr:bound}
For any DDR of the form~\eqref{eq:ddr:rule} and statistics $\calS$,
\begin{align}
    \log_N \sup_{\calD \models \calS}\;
    \min_{(Q_{\bm B})_{\bm B\in\calB}}\;
    \max_{\bm B \in \calB}\; |Q_{\bm B}(\calD)|
    \quad\leq\quad
    \max_{h \models \calS, \Gamma_n}\;
    \min_{\bm B \in \calB}\; h(\bm B).
    \label{eq:ddr:polymatroid:bound}
\end{align}
\end{theorem}

Note that Theorem~\ref{thm:ddr:bound} is the exact generalization of the output size bound for a CQ
from Theorem~\ref{thm:cq:bound}, because a CQ is simply a DDR with $|\calB| = 1$.

\subsection{Costing Adaptive Query Plans}
\label{subsec:costing:adaptive}

Returning to the quantity~\eqref{eq:multi-td:bound:estimate} that we want to estimate,
we replace each inner term $\sup_{\calD \models \calS} \min_{(Q_{\bm B})_{\bm B\in\calB}} \max_{\bm B \in \calB}
|Q_{\bm B}(\calD)|$ by the upper bound from Theorem~\ref{thm:ddr:bound},
and define the \emph{submodular width} of $Q$ with respect to $\calS$ as:
\begin{align}
    \subw(Q, \calS)
    \quad\defeq\quad
    \max_{\calB \in \bs(Q)}\;
    \max_{h \models \calS, \Gamma_n}\;
    \min_{\bm B \in \calB}\; h(\bm B).
    \label{eq:subw:via:bs}
\end{align}
Since this definition is expressed in terms of the bag selector set $\bs(Q)$, which is not
immediately intuitive, we reformulate it.
By swapping the two $\max$ operators, it is not hard to see that:
\begin{align}
    \subw(Q, \calS)
    &=
    \max_{\calB \in \bs(Q)}\;
    \max_{h \models \calS, \Gamma_n}\;
    \min_{\bm B \in \calB}\; h(\bm B)\nonumber\\
    &=
    \max_{h \models \calS, \Gamma_n}\;
    \min_{\calT \in \td(Q)}\;
    \max_{\bm B \in \bags(\calT)}\; h(\bm B).
    \label{eq:subw}
\end{align}
The last expression can equivalently serve as the definition of the submodular width,
and it can be the more natural one to start with; however, the other expression is more
convenient to compute as we shall see in the next section.

The above definition of the submodular width {\em generalizes} the original one by
Marx~\cite{DBLP:journals/jacm/Marx13} in two ways: (1) it applies to arbitrary statistics
$\calS$ about the input database instance $\calD$, whereas the original definition only
considers the input size $N$, and (2) it applies to arbitrary CQs $Q$ thanks to the notion
of {\em free-connex} TDs, whereas the original definition only applies to Boolean CQs.
Marx~\cite{DBLP:journals/jacm/Marx13} described an algorithm that can evaluate a {\em
Boolean} CQ $Q$ under {\em identical cardinality constraints} $\calS$ in time
$O(N^{c\cdot\subw(Q, \calS)})$, where $c > 1$.

For the query $\cycle$ under $\cyclestats$ (Eq.~\eqref{eq:4cycle:stats}), the submodular width becomes:
\begin{align}
    \subw(Q_\square, \calS_\square)
        = \max_{h \models \calS_\square, \Gamma_4}
        \min(\max(h(XYZ), h(ZWX)),\max(h(YZW), h(WXY)))
    \label{eq:subw:4cycle}
\end{align}

\section{Computing the Submodular Width, Shannon-Flow Inequalities}
\label{sec:subw}

Our main goal is to introduce an algorithm (called
$\panda$~\cite{DBLP:conf/pods/Khamis0S17,theoretics:13722,pandaexpress}) that can evaluate
any CQ $Q$ in time $O(N^{\subw(Q, \calS)}\cdot \log N+\OUT)$ for any set of statistics
$\calS$. This runtime remains the best known for {\em any} CQ $Q$ over all combinatorial
algorithms. By comparing Eq.~\eqref{eq:subw} to Eq.~\eqref{eq:fhtw}, note that the min-max
inequality implies that $\subw(Q, \calS) \leq \fhtw(Q, \calS)$ for any $Q$ and $\calS$.
Hence, the same algorithm achieves also a runtime of $O(N^{\fhtw(Q, \calS)}\cdot \log N +
\OUT)$. Before we introduce the $\panda$ algorithm, we first need to understand how to
compute $\subw(Q, \calS)$ itself, and how Shannon inequalities come into play in this
computation.

\subsection{Computing the submodular width via linear programming}
\label{subsec:subw:computation}

From~\eqref{eq:subw},  we can compute $\subw(Q, \calS)$ by solving separately, for each
bag selector $\calB \in \bs(Q)$, the inner optimization problem
\[ \max_{h \models \calS, \Gamma_n}\quad \min_{\bm B \in \calB}\quad h(\bm B), \] which
is the RHS of~\eqref{eq:subw:via:bs}. Recall that, from Theorem~\ref{thm:ddr:bound},
this is an upper bound on the worst-case size of the DDR
corresponding to $\calB$. We shall show that this is just a linear program, via our running
example. In other words, instead of dealing directly with the submodular width, we can
henceforth focus on computing the worst-case size bound of an arbitrary DDR. It will turn
out that computing this bound gives us a hint on how to evaluate the DDR itself in the time
proportional to the worst-case size bound, which means we can evaluate any CQ $Q$ in time
$O(N^{\subw(Q, \calS)}\cdot \log N + \OUT)$.

For example, $\subw(Q_\square, \calS_\square)$ can be written as a max of four expressions,
one for each bag selector: {\allowdisplaybreaks[0]
\begin{align}
    \subw(Q_\square, \calS_\square)
        = \max\biggl(
        &{\color{blue}\overbrace{\color{black}\max_{h \models \calS_\square, \Gamma_4}\min(h(XYZ), h(YZW))}^{\text{Equivalent to the LP from Eq.~\eqref{eq:lp1}}}},\max_{h \models \calS_\square, \Gamma_4}\min(h(XYZ), h(WXY)),\nonumber\\
        &\max_{h \models \calS_\square, \Gamma_4}\min(h(ZWX), h(YZW)),
        \max_{h \models \calS_\square, \Gamma_4}\min(h(ZWX), h(WXY))\biggr)
    \label{eq:subw:4cycle:distributed}
\end{align}}
Consider the first expression inside the outer max above. This expression is
equivalent to the following LP, where the objective is to maximize a {\em new} variable $t$
subject to $t \leq h(XYZ)$ and $t \leq h(YZW)$. At optimality, we must have $t =
\min(h(XYZ), h(YZW))$: {\allowdisplaybreaks[0]
\begin{align}
    \label{eq:lp1}
    \text{max} \quad&t\\
    \text{s.t.}\quad& t \leq h(XYZ), \quad t \leq h(YZW),\label{eq:lp1:c1}\\
                    & h(XY), h(YZ), h(ZW), h(WX) \leq 1,\label{eq:lp1:c2}\\
                    & h \in \Gamma_4
                    \label{eq:lp1:c3}
\end{align}
}
Recall that $h \in \Gamma_4$ is shorthand for saying $h$ satisfies the linear
constraints~\eqref{eq:h:emptyset}-\eqref{eq:h:submodularity} for $\bm V=\{X, Y, Z, W\}$,
and that $\Gamma_4$ (and $\Gamma_n$ in general) is a convex polyhedron.
We will see in the next section that the above LP has an optimal objective value of $3/2$,
and the same holds for the other 3 LPs in Eq.~\eqref{eq:subw:4cycle:distributed}.
Therefore, $\subw(Q_\square, \calS_\square) = 3/2$.

\subsection{Shannon-flow inequalities}
\label{subsec:subw:shannon}

A numeric LP-solution to~\eqref{eq:lp1} does not give us much insight into the structure of
the solution. To gain more insight, we use convex duality: We introduce Lagrange multipliers
$\lambda_1, \lambda_2$ for the two constraints in Eq.~\eqref{eq:lp1:c1}, and $w_1, w_2, w_3,
w_4$ for the four constraints in Eq.~\eqref{eq:lp1:c2}. The Lagrangian dual function is:
\begin{align*}
    & \calL(\lambda_1, \lambda_2, w_1, w_2, w_3, w_4) \\
    &= \max_{t, h \in \Gamma_4}
        \bigl\{t + \lambda_1 (h(XYZ)-t) + \lambda_2 (h(YZW)-t) \bigr.\\
    & \qquad \qquad \bigl. + w_1 (1-h(XY)) + w_2 (1-h(YZ))
    + w_3 (1-h(ZW)) + w_4 (1-h(WX))\bigr\}\\
    &= w_1 + w_2 + w_3 + w_4 + (1-\lambda_1 - \lambda_2) \cdot \max_t t  \\
    &\qquad \qquad \bigl.  + \max_{h \in \Gamma_4} \bigl\{ \lambda_1 h(XYZ) + \lambda_2 h(YZW) - w_1 h(XY) - w_2 h(YZ) - w_3 h(ZW) - w_4 h(WX)\bigr\}
\end{align*}
Let $\opt$ denote the optimal objective value of the primal LP from Eq.~\eqref{eq:lp1}.
Without loss of generality, we can assume that $\opt$ is finite, otherwise the submodular
width is unbounded (when there are insufficient degree constraints).
Due to strong duality, the Lagrangian dual problem has the same objective value as the primal LP,
namely
\begin{align*}
    \opt = \min_{\lambda_1, \lambda_2, w_1, w_2, w_3, w_4 \geq 0} \calL(\lambda_1, \lambda_2, w_1, w_2, w_3, w_4)
\end{align*}
Let $(\lambda^*_1, \lambda^*_2, w^*_1, w^*_2, w^*_3, w^*_4)$ be an optimal solution to the
above dual problem; then, in order for $\opt$ to be finite the following must hold:
\begin{align}
    \lambda_1^* + \lambda_2^* \quad=\quad& 1 \label{eq:lambda:sum}\\
    \lambda_1^* h(XYZ) + \lambda_2^* h(YZW) \quad\leq\quad& w_1^* h(XY) + w_2^* h(YZ) + w_3^* h(ZW) +w_4^* h(WX), \quad\text{for all $h \in \Gamma_4$}
    \label{eqn:shannon:inequality:example}
\end{align}
Equality~\eqref{eq:lambda:sum} holds because if $\lambda_1^* + \lambda_2^* \neq 1$, then the
term $(1-\lambda_1^* - \lambda_2^*) \cdot \max_t t$ can be made arbitrarily large by
changing $t$, which is unconstrained in the primal LP. The
inequality~\eqref{eqn:shannon:inequality:example} holds because if it was violated by some
$h \in \Gamma_4$, then the term $\max_{h \in \Gamma_4} \{\cdots\}$ can be made arbitrarily
large by scaling $h$.
Since~\eqref{eqn:shannon:inequality:example} holds, we have
$\max_{h \in \Gamma_4} \{\cdots\}=0$ because it is achieved by $h=0$ (the zero function is in $\Gamma_4$).
Consequently, the Lagrangian dual problem can be reformulated as:
\begin{align}
\min & \quad w_1 + w_2 + w_3 + w_4 \label{eq:dual:lp}\\
\text{s.t.} & \quad \lambda_1 + \lambda_2 = 1 \\
& \quad \lambda_1 h(XYZ) + \lambda_2 h(YZW) \;\leq\; w_1 h(XY) + w_2 h(YZ) + w_3 h(ZW) + w_4 h(WX) \quad\text{for all $h \in \Gamma_4$}
\label{eqn:shannon:flow:inequality:example}\\
& \quad \lambda_1, \lambda_2, w_1, w_2, w_3, w_4 \geq 0
\end{align}
Inequality~\eqref{eqn:shannon:flow:inequality:example} is called a {\em Shannon-flow
inequality}, for given parameters $\bm \lambda, \bm w$. In words, we were able to
reformulate the optimization problem as finding non-negative coefficients $\bm \lambda, \bm
w$ such that $\norm{\bm\lambda}_1 = 1$ and the Shannon-flow inequality holds for all $h \in
\Gamma_4$, and such that the sum of the $w_i$'s is minimized.

For example, an optimal dual solution corresponds to the following Shannon-flow inequality,
where $\lambda_1 = \lambda_2 = w_1=w_2=w_3=1/2$ and $w_4=0$:
\begin{align}
    \frac{1}{2} h(XYZ) + \frac{1}{2} h(YZW)
        \quad\leq\quad \frac{1}{2} h(XY) + \frac{1}{2} h(YZ) + \frac{1}{2} h(ZW)
    \label{eq:shannon:example}
\end{align}
The above is a Shannon inequality because it is a sum of half of the following
submodularities (Eq.~\eqref{eq:h:submodularity}):
\begin{align}
    -h(XY) - h(YZ) + h(Y) + h(XYZ) &\leq 0& \text{(submodularity)}\label{eq:shannon:example:1}\\
    -h(Y) - h(ZW) +  h(YZW) &\leq 0&\text{(submodularity)}\label{eq:shannon:example:2}
\end{align}

Generalizing the above reasoning to an arbitrary CQ $Q$ and statistics $\calS$, we can formally show
this reformulation:

\begin{lemma}[\cite{theoretics:13722,pandaexpress}]
\label{lmm:ddr:lp}
Let $\calS$ be input statistics of the DDR~\eqref{eq:ddr:rule} over variables $\bm V$. Then,
$\max_{\bm h \models \calS} \min_{\bm B \in \calB} h(\bm B)$ has exactly the same optimal
objective value as the following optimization problem:
\begin{align}
    \min_{\bm \lambda, \bm w} \quad &
    \sum_{N_{\bm Y|\bm X} \in \calS} w_{\bm Y|\bm X} \log_N N_{\bm Y|\bm X} \\
    \text{s.t.} \quad &
    \sum_{\bm B \in \calB}\lambda_{\bm B}\cdot h(\bm B)
    \quad\leq\quad \sum_{N_{\bm Y|\bm X} \in \calS} w_{\bm Y|\bm X} \cdot h(\bm Y|\bm X)
    ,\quad\text{$\forall h \in \Gamma_n$} \label{eqn:shannon:general}\\
    & \norm{\bm \lambda}_1 = 1 \text{ and } \bm \lambda \geq \bm 0, \bm w \geq \bm 0, \nonumber
\end{align}
provided that one of them is bounded.
Moreover, both problems can be solved with a linear program.
\end{lemma}


The Shannon-flow inequality~\eqref{eqn:shannon:general}
will play a crucial role in the design of the $\panda$ algorithm, as we will see in the next section.
The coefficients $\lambda_{\bm B}, w_{\bm Y|\bm X} \geq 0$ can be assumed to be rational numbers,
We refer to such inequalities as {\em rational} Shannon-flow inequalities.
As a special case, when the DDR is a CQ $Q$ and all statistics in $\calS$ are cardinality
constraints of the form $h(\bm X) \leq \log_N N_{\bm X}$ for each relation $R(\bm X) \in \atoms(Q)$,
the Shannon-flow inequality~\eqref{eqn:shannon:general} reduces to {\em Shearer's
lemma}~\cite{MR859293}.

The significance of Shannon-flow inequalities is that they can be used to upper-bound the
worst-case size of DDRs, and thus the cost of evaluating CQs using DDRs.
\begin{thm}[\cite{theoretics:13722,pandaexpress}]
\label{thm:ddr:bound:shannon}
Given a DDR of the form~\eqref{eq:ddr:rule}, input statistics $\calS$,
 and non-negative coefficients $\bm \lambda, \bm w$
with $\norm{\bm\lambda}_1 = 1$ forming a Shannon-flow inequality of the form~\eqref{eqn:shannon:general},
then
\begin{align}
\sup_{\calD \models \calS}\; \min_{(Q_{\bm B})_{\bm B\in\calB}}\;\max_{\bm B \in \calB}\; |Q(\bm B)|
    &\leq \prod_{N_{\bm Y|\bm X} \in \calS} N_{\bm Y|\bm X}^{w_{\bm Y|\bm X}}.
\end{align}
where the $\min$ is taken over all feasible solutions $Q_{\bm B}$ to the
DDR~\eqref{eq:ddr:rule}.
\end{thm}
\begin{proof}
From the fact that~\eqref{eqn:shannon:general} is a Shannon-flow inequality, we can upper-bound the RHS
of~\eqref{eq:ddr:polymatroid:bound} from
Theorem~\ref{thm:ddr:bound}: for any $h \models \calS, \Gamma_n$,
\begin{align*}
    \min_{\bm B \in \calB} h(\bm B)
    \;\leq\; \sum_{\bm B \in \calB} \lambda_{\bm B}\, h(\bm B)
    \;\leq\; \sum_{N_{\bm Y|\bm X} \in \calS} w_{\bm Y|\bm X}\, h(\bm Y|\bm X)\;\leq\; \sum_{N_{\bm Y|\bm X} \in \calS} w_{\bm Y|\bm X} \log_N N_{\bm Y|\bm X}.
\end{align*}
Hence, $$\max_{h \models \calS, \Gamma_n} \min_{\bm B \in \calB} h(\bm B) \;\leq\; \sum_{N_{\bm
Y|\bm X} \in \calS} w_{\bm Y|\bm X} \log_N N_{\bm Y|\bm X},$$ and the worst-case size of the
DDR is at most $\prod_{N_{\bm Y|\bm X} \in \calS} N_{\bm Y|\bm X}^{w_{\bm Y|\bm X}}.$
\end{proof}

For example, we have shown that~\eqref{eq:shannon:example} is a Shannon inequality,
where $\lambda_1 = \lambda_2 = w_1=w_2=w_3=1/2$ and $w_4=0$. Hence, for the
DDR~\eqref{eq:4cycle:ddr:example} with statistics $\calS_\square$, the worst-case size of
the DDR is
\begin{align}
    \min_{A_{12},A_{21}}\max \{ |A_{11}(X,Y,Z)|, |A_{21}(Y,Z,W)| \} \;\leq\; \sqrt{N_{XY} N_{YZ} N_{ZW}} \;=\; N^{3/2},
    \label{eq:ddr:example:size:bound}
\end{align}
where the $\min$ is over all feasible solutions $A_{12}, A_{21}$ to the DDR~\eqref{eq:4cycle:ddr:example}.
In this example, it can also be shown that this bound is tight, i.e., equal to the $\min$ over all feasible solutions $A_{12}, A_{21}$.

\section{Proof Sequences of Shannon-Flow Inequalities}
\label{sec:proof:sequences}

This section develops tools for proving that~\eqref{eqn:shannon:general} holds for all
polymatroids; the next section shows how to turn such a proof into an algorithm.
A {\em rational} Shannon-flow inequality is a non-negative rational combination of the basic
Shannon inequalities~\eqref{eq:h:emptyset}-\eqref{eq:h:submodularity}; an {\em integral}
one uses integer coefficients. Every rational Shannon-flow inequality can be converted to an
integral one by multiplying with a common denominator. For example, Eq.~\eqref{eq:shannon:example} is rational
--- it is the sum of $1/2$ of each of the submodularities~\eqref{eq:shannon:example:1}
and~\eqref{eq:shannon:example:2} --- and multiplying by 2 gives its integral form:
\begin{align}
    h(XYZ) + h(YZW)
        \quad\leq\quad h(XY) + h(YZ) + h(ZW)
    \label{eq:shannon:example:integral}
\end{align}
Eq.~\eqref{eq:shannon:example:integral} is just the sum of the two submodularities
from Eq.~\eqref{eq:shannon:example:1} and~\eqref{eq:shannon:example:2}.
Equivalently, the following is an {\em algebraic identity} (holding for all values of $h(\cdot)$):
{\allowdisplaybreaks[0]
\begin{align}
    \label{eq:shannon:example:identity}
    {\color{blue}\overbrace{\color{black}h(XYZ) + h(YZW)}^{\text{``target terms''}}} \quad=\quad& {\color{blue}\overbrace{\color{black}h(XY) + h(YZ) + h(ZW)}^{\text{``source terms''}}}\\
        &+ {\color{blue}\underbrace{\color{black}(-h(XY) - h(YZ) + h(Y) + h(XYZ))}_{\color{blue}\text{($\leq 0$ by submodularity)}}}\nonumber\\
        &+ {\color{blue}\underbrace{\color{black}(-h(Y) - h(ZW) +  h(YZW))}_{\text{($\leq 0$ by submodularity)}}}\nonumber
\end{align}
} We call~\eqref{eq:shannon:example:identity} the {\em identity form}
of~\eqref{eq:shannon:example:integral}: the LHS is the {\em target terms} (LHS
of~\eqref{eq:shannon:example:integral}), and the RHS is the {\em source terms} (RHS
of~\eqref{eq:shannon:example:integral}) plus submodularities (and possibly monotonicities).
Every integral Shannon-flow inequality has such an identity form, and this structure will be
key to the lemmas below.

\subsection{Proof Sequence Construction}
\label{subsec:proof-sequence}

Our first lemma says that every integral Shannon-flow
inequality~\eqref{eq:shannon:example:integral} can be proved through a sequence of steps
that transform the RHS into the LHS. The proof steps are of the following kinds, for some
sets of variables $\bm X, \bm Y, \bm Z$:
\begin{align}
    \text{Decomposition Step:}\quad& h(\bm X\bm Y) \to h(\bm X) + h(\bm Y|\bm X)\label{eq:proof:decomposition}\\
    \text{Composition Step:}\quad& h(\bm X) + h(\bm Y|\bm X) \to h(\bm X \bm Y)\label{eq:proof:composition}\\
    \text{Monotonicity Step:}\quad& h(\bm X\bm Y) \to h(\bm X)\label{eq:proof:monotonicity}\\
    \text{Submodularity Step:}\quad& h(\bm Y|\bm X) \to h(\bm Y|\bm X\bm Z)\label{eq:proof:submodularity}
\end{align}
Note that by basic Shannon inequalities, each one of the above steps replaces one or two
terms with one or two {\em smaller} terms. In particular, for composition and decomposition
steps, $h(\bm X \bm Y) = h(\bm X) + h(\bm Y|\bm X)$ by definition of $h(\bm Y|\bm X)$. For
monotonicity steps, $h(\bm X\bm Y) \geq h(\bm X)$ by Eq.~\eqref{eq:h:monotonicity}, whereas
for submodularity steps, $h(\bm Y|\bm X) \geq h(\bm Y|\bm X\bm Z)$ by
Eq.~\eqref{eq:h:submodularity:equiv}.

As an example, the right column of Table~\ref{table:ps} shows a proof sequence that proves
the integral Shannon inequality~\eqref{eq:shannon:example:integral}. One can verify that
applying these steps one by one transforms the RHS of~\eqref{eq:shannon:example:integral}
into the LHS. But now the question is: How do we constructively prove the existence of such
a proof sequence for any integral Shannon inequality?

To that end, we take the identity form, which is Eq.~\eqref{eq:shannon:example:identity} in
our example. Identity~\eqref{eq:shannon:example:identity} has three {\em
unconditional}~\footnote{Recall that a term $h(\bm Y|\bm X)$ is called {\em unconditional}
iff $\bm X = \emptyset$. For example, $h(XY)$ is unconditional, whereas $h(Y|X)$ is not.}
source terms and two target terms\footnote{Target terms are always unconditional because the
underlying Shannon inequality always has the form of Eq.~\eqref{eqn:shannon:general}, where
all terms on the LHS have the form $\lambda_{\bm B}h(\bm B)$.}. It is not a coincidence that
the number of unconditional source terms is greater than or equal to the number of target
terms.
This is because the polymatroid $h(\bm X) \defeq 1$ for all $\bm X \neq \emptyset$ must satisfy the inequality:  then the LHS is the number of targets, and the RHS is {\em at most}\footnote{Note that the second submodularity $-h(Y) - h(ZW) +  h(YZW)$ in Eq.~\eqref{eq:shannon:example:identity} contributes $-1$ to the RHS, hence we say ``at most''.} the number of unconditional source terms, and the claim follows immediately.
This claim can be quite useful: {\em As long as we sill have at
least one target term, we must have at least one unconditional source term.}

Pick an arbitrary unconditional source term in Identity~\eqref{eq:shannon:example:identity}.
Let's say we pick $h(YZ)$, as shown in the first row of Table~\ref{table:ps}. Note that
$h(YZ)$ is not a target term. Therefore, $h(YZ)$ must be getting {\em cancelled} by some
term on the RHS of {\em Identity}~\eqref{eq:shannon:example:identity}. In our example,
$h(YZ)$ is getting cancelled by the term $-h(YZ)$ from the first submodularity. To get rid
of the term $h(YZ)$ while maintaining the identity, we apply two proof steps:
\begin{align*}
    h(YZ) &\to h(Y) + h(Z|Y) &\text{(decomposition step)}\\
    h(Z|Y) &\to h(Z|XY) &\text{(submodularity step)}
\end{align*}
We also drop the submodularity that was used to cancel $h(YZ)$.
One can verify the resulting identity, depicted in the second row of Table~\ref{table:ps},
is indeed a valid identity.

Now, we pick another unconditional source term from the new identity, say $h(XY)$. Since
$h(XY)$ is not a target term, it must be getting cancelled by some term on the RHS. In this
case, $h(XY)$ is not getting cancelled by a submodularity term, but rather by the
conditional $h(Z|XY)$. We apply a single composition step $h(XY) + h(Z|XY) \to h(XYZ)$, thus
resulting in another valid identity, shown in the third row of Table~\ref{table:ps}.

We pick another unconditional source term from the new identity, say $h(XYZ)$. This time,
$h(XYZ)$ is a target term, which allows us to cancel it from both sides. We continue this
process until we reach the trivial identity $0=0$. This proof sequence construction works on
the identity form of {\em any} integral Shannon inequality~\eqref{eqn:shannon:general}.

\begin{table*}
\[
\small
\begin{array}{|c|c|}
    \hline
    \rowcolor{lightgray}\text{\bf Identity} & \text{\bf Proof Steps}\\\hline
    \begin{aligned}[c]
        h(XYZ) + h(YZW) =& h(XY) \positiveterm{$+h(YZ)$} + h(ZW)\\
            &+ (-h(XY) \negativeterm{$- h(YZ)$} + h(Y) + h(XYZ))\\
            &+ (-h(Y) - h(ZW) +  h(YZW))
    \end{aligned}&
    \begin{aligned}[c]
        \positiveterm{$h(YZ)$} &\to h(Y) + h(Z|Y)\\
        h(Z|Y) &\to h(Z|XY)
    \end{aligned}\\\hline
    \begin{aligned}[c]
        h(XYZ) + h(YZW) =& \positiveterm{$h(XY)$} + h(Y) + \negativeterm{$h(Z|XY)$} + h(ZW)\\
            &+ (-h(Y) - h(ZW) +  h(YZW))
    \end{aligned}&
    \begin{aligned}[c]
        \positiveterm{$h(XY)$} + \negativeterm{$h(Z|XY)$} &\to h(XYZ)
    \end{aligned}\\\hline
        \begin{aligned}[c]
        \negativeterm{$h(XYZ)$} + h(YZW) =& \positiveterm{$h(XYZ)$} + h(Y) + h(ZW)\\
            &+ (-h(Y) - h(ZW) +  h(YZW))
    \end{aligned}&
    \begin{aligned}[c]
        \text{--}
    \end{aligned}\\\hline
    \begin{aligned}[c]
        h(YZW) =& \positiveterm{$h(Y)$} + h(ZW)\\
            &+ (\negativeterm{$-h(Y)$} - h(ZW) +  h(YZW))
    \end{aligned}&
    \begin{aligned}[c]
        \positiveterm{$h(Y)$} &\to h(Y|ZW)
    \end{aligned}\\\hline
    \begin{aligned}[c]
        h(YZW) =& \negativeterm{$h(Y|ZW)$} + \positiveterm{$h(ZW)$}
    \end{aligned}&
    \begin{aligned}[c]
        \positiveterm{$h(ZW)$} + \negativeterm{$h(Y|ZW)$}&\to h(YZW)
    \end{aligned}\\\hline
    \begin{aligned}[c]
        \negativeterm{$h(YZW)$} =& \positiveterm{$h(YZW)$}
    \end{aligned}&
    \begin{aligned}[c]
        \text{--}
    \end{aligned}\\\hline
    \begin{aligned}[c]
        0 = 0
    \end{aligned}&
    \begin{aligned}[c]
        \text{--}
    \end{aligned}\\\hline
\end{array}
\]
\caption{Proof sequence construction for the Identity~\eqref{eq:shannon:example:identity}.
    The left column shows a sequence of identities, starting from~\eqref{eq:shannon:example:identity},
    while the right column shows the proof steps that transform each identity to the next.
    As long as the identity still has target terms on the LHS, it must still have
    at least one {\em unconditional} source term on the RHS.
    We repeatedly pick one \positiveterm{source term}
    and look for a \negativeterm{term} that cancels it.
    Depending on the kind of \negativeterm{term} we find, we apply some proof steps
    to get a new ``simpler'' identity.
    We keep doing this until we reach the trivial identity $0=0$.}
\label{table:ps}
\end{table*}

\subsection{Reset Lemma}
\label{subsec:reset}
Our second lemma, called the {\em Reset Lemma}, says the following: Given any integral
Shannon inequality of the form~\eqref{eqn:shannon:general}, let $h(\bm Y|\emptyset)$ be any
{\em unconditional} term  on the RHS. Then, we can always drop $h(\bm Y|\emptyset)$ from the
RHS and lose {\em at most} one term $h(\bm B)$ from the LHS, in order to obtain another
valid integral Shannon inequality.\footnote{ In the process, we might lose more terms from
the RHS as well.} For example, consider the Shannon inequality from
Eq.~\eqref{eq:shannon:example:integral}, and suppose we want to drop the unconditional term
$h(XY)$ from the RHS. Then, the lemma promises that this can only cost us losing one of the
two terms $h(XYZ)$ or $h(YZW)$ from the LHS, but never both! Indeed, in this case, this is
one valid Shannon inequality that we can obtain:
\begin{align}
    h(YZW) \quad\leq\quad h(YZ) + h(ZW)
    \label{eq:shannon:example:reset}
\end{align}
In particular, the above is a valid Shannon inequality because it is a sum of the following
monotonicity and submodularity:
\begin{align*}
    -h(YZ)+h(Y) &\quad\leq\quad 0 &\text{(monotonicity)}\\
    -h(Y) - h(ZW) +  h(YZW) &\quad\leq\quad 0 &\text{(submodularity)}
\end{align*}
But how do we systematically prove the existence of such a Shannon
inequality~\eqref{eq:shannon:example:reset}? To that end, we start from the identity
form~\eqref{eq:shannon:example:identity} of the original Shannon
inequality~\eqref{eq:shannon:example:integral}. The proof is somewhat similar to the proof
sequence construction. In particular, consider the source term $h(XY)$ on the RHS of
Identity~\eqref{eq:shannon:example:identity}. Since $h(XY)$ is not a target term, it must be
getting cancelled by some term on the RHS, which in this case is the term $-h(XY)$ from the
first submodularity in~\eqref{eq:shannon:example:identity}. We replace $h(XY)$ with $h(XYZ)$
and replace the first submodularity with the monotonicity $-h(YZ) + h(Y) \leq 0$, thus
resulting in the following identity: {\allowdisplaybreaks[0]
\begin{align*}
    h(XYZ) + h(YZW) =& h(XYZ) + h(YZ) + h(ZW)\\
        & + {\color{blue}\underbrace{\color{black}(-h(YZ) + h(Y))}_{\text{($\leq 0$ by monotonicity)}}}\nonumber\\
        &+ {\color{blue}\underbrace{\color{black}(-h(Y) - h(ZW) +  h(YZW))}_{\text{($\leq 0$ by submodularity)}}}\nonumber
\end{align*}
} One can verify that the above is indeed a valid identity. Instead of dropping $h(XY)$ from
the RHS of the original identity~\eqref{eq:shannon:example:identity}, our target now is to
drop $h(XYZ)$ from the RHS of this new identity above. We continue this process inductively.
In this case, we are lucky because $h(XYZ)$ is a target term, hence we drop it from both
sides, thus obtaining~\eqref{eq:shannon:example:reset}. This inductive proof strategy can
prove the Reset lemma for {\em any} integral Shannon
inequality~\eqref{eqn:shannon:general}~\cite{theoretics:13722}.

\section{Evaluating DDRs using the $\panda$ Algorithm}
\label{sec:panda}


Given a CQ $Q$, a set of statistics $\calS$, and a database instance $\calD\models\calS$, we
saw in Section~\ref{sec:adaptive:query:plans} how to reduce evaluating $Q$ over $\calD$ to evaluating a set
of DDRs, each of which has an output size bound of $N^{\subw(Q, \calS)}$. In this section,
we present the $\panda$ algorithm for evaluating any DDR in time proportional to its output
size bound (with an extra $\log N$-factor). This in turn allows us to evaluate the original
CQ $Q$ in time $O(N^{\subw(Q, \calS)}\cdot \log N+\OUT)$, as desired.
The original $\panda$ algorithm~\cite{DBLP:conf/pods/Khamis0S17,theoretics:13722} had an extra $\polylog(N)$ factor in the runtime,
but an improved version, called $\pandaexpress$~\cite{pandaexpress}, was recently proposed
that replaces this factor with a single $\log N$ factor.\footnote{Another algorithm~\cite{jaguar} replaces the $\polylog(N)$ factor with $N^{\epsilon}$ where $\epsilon$ can be arbitrarily small.}
We present here the $\pandaexpress$ version through an example, and leave the details to~\cite{pandaexpress}.

Our $\pandaexpress$ algorithm for DDRs relies on the lemmas concerning Shannon inequalities from
Section~\ref{sec:proof:sequences}. In order to explain the algorithm, we start in
Section~\ref{subsec:panda:ddr:proof} by proposing a new proof~\cite{pandaexpress} of the
output size bound of a DDR (Theorem~\ref{thm:ddr:bound:shannon}). Then, in
Section~\ref{subsec:panda:ddr:algorithm}, we show how to turn this proof into an actual
algorithm for evaluating the DDR.

\subsection{A New Proof of the Output Size Bound of a DDR}
\label{subsec:panda:ddr:proof}

We give an alternative proof of Theorem~\ref{thm:ddr:bound:shannon} that is based on {\em
sub-probability measures}. To that end, we start with a brief overview of sub-probability
measures. A {\em sub-probability measure} is just like a probability distribution, except
that we only require the total mass to be {\em at most} 1 (instead of exactly 1). In
particular, given a set of variables $\bm X$, a sub-probability measure over $\bm X$ is a
function $p_{\bm X}:\dom^{\bm X} \to [0, 1]$ such that $\sum_{\bm x\in\dom^{\bm X}} p_{\bm
X}(\bm x) \leq 1$. Similarly, we define $p_{\bm Y|\bm X=\bm x}(\bm y)$ to be a {\em
conditional} sub-probability measure, i.e.,~ a sub-probability measure over $\bm Y$ for
every value $\bm X =\bm x$, and we abbreviate it as $p_{\bm Y|\bm X}(\bm y|\bm x)$. Although
sub-probability measures are not real probability distributions, we extend the standard
definitions of {\em support}, {\em marginal} and {\em conditional} to them. In particular,
given a sub-probability measure $p_{\bm X\bm Y}$ over $\bm X\bm Y$, we define its {\em
marginal} and {\em conditional} on $\bm X$, respectively as follows:
\begin{align*}
    p_{\bm X}(\bm x) \defeq \sum_{\bm y} p_{\bm X\bm Y}(\bm x, \bm y), \quad\quad\quad
    p_{\bm Y|\bm X}(\bm y|\bm x) \defeq
        \frac{p_{\bm X\bm Y}(\bm x, \bm y)}{p_{\bm X}(\bm x)}
\end{align*}
Just like for probability distributions, {\em if all non-zero probabilities in a
sub-probability measure are at least $1/B$ for some constant $B>0$, then the number of
non-zero probabilities (i.e.~the support size) is at most $B$.}

To demonstrate the new proof of Theorem~\ref{thm:ddr:bound:shannon}, we re-prove the
bound~\eqref{eq:ddr:example:size:bound} using sub-probability measures. In particular, we
take the DDR from Eq.~\eqref{eq:4cycle:ddr:example} and show that the Shannon inequality
from Eq.~\eqref{eq:shannon:example:integral} implies an output size bound of $B \defeq
N^{3/2}$ for this DDR. (Recall that this Shannon inequality came from an optimal dual
solution to the LP in Eq.~\eqref{eq:lp1}.) To that end, we use the proof sequence of
Inequality~\eqref{eq:shannon:example:integral} from Table~\ref{table:ps}. This proof
sequence transforms the RHS of Inequality~\eqref{eq:shannon:example:integral} into the LHS.
We associate each source term on the RHS of Inequality~\eqref{eq:shannon:example:integral}
with a sub-probability measure. Every time we perform a proof step transforming one or two
entropy terms into new terms, we construct new sub-probability measures to associate with
the new terms, thus keeping the RHS of the Shannon inequality in lockstep with the
sub-probability measures. The process is depicted in Table~\ref{table:sub:prob}.

In particular, we start by defining three sub-probability measures $p_{XY}(x, y)$,
$p_{YZ}(y, z)$, and $p_{ZW}(z, w)$ to associate with the three source terms $h(XY)$,
$h(YZ)$, and $h(ZW)$ on the RHS of Inequality~\eqref{eq:shannon:example:integral}. All three
are defined to be uniform over the supports of the underlying input relations $R(X, Y)$,
$S(Y, Z)$, and $T(Z, W)$, respectively. Note that for every tuple $(x, y, z, w)$ in the
join, $\Join_\square$, of the input relations, the following inequality holds:
\begin{align}
    p_{XY}(x, y) \cdot p_{YZ}(y, z) \cdot p_{ZW}(z, w) \quad\geq\quad \frac{1}{N^3}=\frac{1}{B^2}
    \label{eq:sub:prob:inequality}
\end{align}
The first proof step (second row of Table~\ref{table:sub:prob}) transforms $h(YZ)$ into
$h(Y) + h(Z|Y)$. The corresponding step in the language of sub-probability measures is
replacing $p_{YZ}(y, z)$ in inequality~\eqref{eq:sub:prob:inequality} with the product of
its marginal $p_Y(y)$ and conditional $p_{Z|Y}(z|y)$. Note that this replacement preserves
the inequality.

The second step transforms $h(Z|Y)$ into $h(Z|XY)$. The corresponding step is simply to
replace $p_{Z|Y}(z|y)$ with $p_{Z|XY}(z|xy)\defeq p_{Z|Y}(z|y)$ in
Inequality~\eqref{eq:sub:prob:inequality}, thus leaving the inequality intact.
The third
step transforms $h(XY) + h(Z|XY)$ into $h(XYZ)$. The corresponding step is to replace
$p_{XY}(x, y) \cdot p_{Z|XY}(z|xy)$ with $p_{XYZ}(x, y, z)$ which is defined as their
product. This again preserves Inequality~\eqref{eq:sub:prob:inequality}.

When we reach the end of the proof sequence, we obtain two sub-probability measures
$p_{XYZ}(x, y, z)$ and $p_{YZW}(y, z, w)$, and inequality~\eqref{eq:sub:prob:inequality}
becomes:
\begin{align}
    p_{XYZ}(x, y, z) \cdot p_{YZW}(y, z, w) \quad\geq\quad \frac{1}{N^3}=\frac{1}{B^2},
    \label{eq:sub:prob:inequality:final}
\end{align}
which holds for all tuples $(x, y, z, w)\in\Join_\square$.
Inequality~\eqref{eq:sub:prob:inequality:final} says that the geometric mean of $p_{XYZ}(x,
y, z)$ and $p_{YZW}(y, z, w)$ must be at least $1/B$. Therefore, for every tuple $(x, y, z,
w)\in\Join_\square$, at least one of $p_{XYZ}(x, y, z)$ and $p_{YZW}(y, z, w)$ must be $\geq
1/B$. We define $\ov p_{XYZ}$ and $\ov p_{YZW}$ to be ``truncated'' versions of $p_{XYZ}$
and $p_{YZW}$, respectively, where we only keep tuples with probability at least $1/B$. And
now, we construct output relations $A_{11}'(X, Y, Z)$ and $A_{21}'(Y, Z, W)$ to the DDR from
Eq.~\eqref{eq:4cycle:ddr:example} by taking the {\em supports} of $\ov p_{XYZ}$ and $\ov p_{YZW}$,
respectively. One can verify that $(A_{11}', A_{21}')$ is indeed a valid output to the DDR.
Moreover, since $\ov p_{XYZ}$ and $\ov p_{YZW}$ are sub-probability measures where all
non-zero probabilities are at least $1/B$, their supports cannot have sizes larger than $B$.
Hence, the constructed output $(A_{11}', A_{21}')$ has size at most $B$, thus proving
the bound~\eqref{eq:ddr:example:size:bound}.

\begin{table*}
\[
\begin{array}{|c|c|}
    \hline
    \rowcolor{lightgray} \text{\bf Shannon Inequality and Steps} & \text{\bf Probabilistic Inequality and Steps} \\\hline
    \rowcolor{lightergray}
    \small
    \begin{multlined}[t]
    h(XYZ)+h(YZW)\leq\\ {\color{blue}h(XY)+h(YZ)+h(ZW) \leq 3 \log_N N}
    \end{multlined}&
    \begin{aligned}[t]
       & \color{red}p_{XY}(x, y) \cdot p_{YZ}(y, z) {\color{red}\cdot p_{ZW}(z, w) \geq 1/N^3}\\
        p_{XY}(x, y) &\defeq 1/N,\quad p_{YZ}(y, z) \defeq 1/N,\quad p_{ZW}(z, w) \defeq 1/N
    \end{aligned}\\\hline
    \begin{aligned}[t]
    \color{blue}{h(YZ) \quad\to\quad h(Y) + h(Z |Y)}
    \end{aligned}&
    \begin{aligned}[t]
        \color{red}p_{YZ}(y, z) &\color{red}\to\quad p_Y(y) \cdot p_{Z|Y}(z|y)\\
        p_Y(y) &\defeq \sum_{z} p_{YZ}(y, z)=\frac{\deg_S(Z|Y=y)}{N}\\
        p_{Z|Y}(z|y) &\defeq \frac{p_{YZ}(y, z)}{p_Y(y)} = \frac{1}{\deg_S(Z|Y=y)}
    \end{aligned}\\\hline
    \begin{aligned}[t]
        \color{blue}h(Z|Y) \quad\to\quad h(Z|XY)
    \end{aligned}&
    \begin{aligned}[t]
        \color{red} p_{Z|Y}(z|y) & \color{red}\to\quad p_{Z|XY}(z|xy)\\
        p_{Z|XY}(z|xy) &\defeq p_{Z|Y}(z|y) \cdot \frac{1}{\deg_S(Z|Y=y)}
    \end{aligned}\\\hline
    \begin{aligned}[t]
        \color{blue}h(XY)+h(Z|XY)\to h(XYZ)
    \end{aligned}&
    \begin{aligned}[t]
        \color{red}p_{XY}(x, y) &\color{red}\cdot p_{Z|XY}(z|xy) \to\quad p_{XYZ}(x, y, z)\\
        p_{XYZ}(x, y, z) &\defeq p_{XY}(x, y) \cdot p_{Z|XY}(z|xy) = \frac{1}{N \cdot \deg_S(Z|Y=y)}
    \end{aligned}\\\hline
    \begin{aligned}[t]
        \color{blue}h(Y) \quad\to\quad h(Y|ZW)
    \end{aligned}&
    \begin{aligned}[t]
        \color{red}p_Y(y) &\color{red}\to\quad p_{Y|ZW}(y|zw)\\
        p_{Y|ZW}(y|zw) &\defeq p_Y(y) = \frac{\deg_S(Z|Y=y)}{N}
    \end{aligned}\\\hline
    \begin{aligned}[t]
        \color{blue}h(ZW) + h(Y|ZW)\to h(Y|ZW)
    \end{aligned}&
    \begin{aligned}[t]
        \color{red}p_{ZW}(z, w) &\color{red}\cdot p_{Y|ZW}(y|zw) \to\quad p_{YZW}(y, z, w)\\
        p_{YZW}(y, z, w) &\defeq p_{ZW}(z, w) \cdot p_{Y|ZW}(y|zw) = \frac{\deg_S(Z|Y=y)}{N^2}
    \end{aligned}\\\hline
    \rowcolor{lightergray}
    \small
    \begin{multlined}
        h(XYZ)+h(YZW)\leq\\ {\color{blue}h(XYZ)+h(YZW) \leq 3 \log_N N}
    \end{multlined}&
    \begin{aligned}
        {\color{red}p_{XYZ}(x, y, z)\cdot p_{YZW}(y, z, w) \geq 1/N^3}
    \end{aligned}\\\hline
\end{array}
\]
\caption{A new proof of the output size bound of the DDR from Eq.~\eqref{eq:4cycle:ddr:example}
using sub-probability measures.
The left column shows the proof steps from Table~\ref{table:ps}
for the Shannon inequality~\eqref{eq:shannon:example:integral},
whereas the right column shows the corresponding steps in the world of sub-probability measures.}
\label{table:sub:prob}
\end{table*}

\subsection{The $\pandaexpress$ Algorithm for DDRs}
\label{subsec:panda:ddr:algorithm}
The size bound proof from Section~\ref{subsec:panda:ddr:proof} can {\em almost} be
immediately turned into an actual algorithm for evaluating the DDR from
Eq.~\eqref{eq:4cycle:ddr:example} in time $O(B) = O(N^{3/2})$. To that end, we just follow
the proof steps from Table~\ref{table:sub:prob}, and compute the sub-probability measures
one by one, with a {\em small twist}: Instead of computing the full sub-probability measures
$p_{XYZ}$ and $p_{YZW}$, we only compute their ``truncated'' versions $\ov p_{XYZ}$ and $\ov
p_{YZW}$ where we only keep tuples with probability at least $1/B$. From
Table~\ref{table:sub:prob}, note that tuples $(x, y, z, w)\in\Join_\square$ where
$p_{XYZ}(x, y, z) \geq 1/B$ are exactly those where $\deg_S(Z|Y=y) \leq N^{1/2}$. In
contrast, tuples where $p_{YZW} \geq 1/B$ are exactly those where $\deg_S(Z|Y=y) \geq
N^{1/2}$. In particular, one can think of the algorithm in this special case as partitioning
the relations $S(Y, Z)$ into ``light'' and ``heavy'' tuples based on the degree of $Y$ in
$S$ being at most or at least $N^{1/2}$, respectively. For light $Y$ tuples, we compute the
join with $R(X, Y)$ to produce the output $A_{11}'(X, Y, Z)$ of the DDR in time
$O(N^{3/2})$. The number of heavy $Y$-values cannot exceed $N^{1/2}$, and the algorithm
takes the Cartesian product of these heavy $Y$-values with $T(Z, W)$ to produce the output
$A_{21}'(Y, Z, W)$ of the DDR also in time $O(N^{3/2})$. The fact that at least one
probability term on the LHS of Inequality~\eqref{eq:sub:prob:inequality:final} has to be at
least $1/B$ is not a coincidence: In general, whenever one probability term in the
inequality drops below $1/B$, the $\pandaexpress$ algorithm~\cite{pandaexpress} uses the Reset
lemma (Section~\ref{subsec:reset}) to drop that term from the inequality while ensuring that
the remaining terms still have a geometric mean of at least $1/B$ (hence one of them must be
at least $1/B$).

In general~\cite{pandaexpress}, $\pandaexpress$ uses sub-probability measures to find a data
partitioning strategy that evaluates any DDR in time proportional to its output size bound.
Beyond threshold-based partitioning (like $\deg_S(Z|Y=y) \leq N^{1/2}$ above), the
algorithm may also compare two degrees against each other via {\em hyperplane
partitioning}~\cite{pandaexpress}, and may partition not just input relations but also
{\em intermediate} ones computed along the way.

\section{Extensions}
\label{sec:extensions}
We now discuss several extensions of the submodular width and the $\panda$ framework along different dimensions.

\subsection{Other Classes of Queries: \scq, \faq}
\label{sec:extensions:counting}

Both the definition of the submodular width from Eq.~\eqref{eq:subw} and the associated
$\panda$ algorithm we have seen are defined for CQs. But now we ask: What about other
classes of queries, for example, the {\em counting version} of CQs
(\scq)~\cite{10.1145/2448496.2448508}, where we want to count satisfying assignments to CQs?
To be more general, we consider aggregate queries over {\em semirings}, which are special
cases of {\em functional aggregate queries (\faqs)}~\cite{DBLP:conf/pods/KhamisNR16}. In
such queries, we are given a semiring $(\bm K, \oplus, \otimes)$ and a {\em $\bm
K$-annotated} database instance $\calD$, i.e., where each tuple in each relation $R$ is
annotated with a value from $\bm K$~\cite{provenance-semirings}. The goal is to compute a
sum of a product of the input relations' annotations over the semiring. For example, the
following is the {semiring version} of $\cycle$ from Eq.~\eqref{eq:4cycle}, for any semiring
$(\bm K, \oplus, \otimes)$:
\begin{align}
    \cycle^{\bm K}(X, Y) &\quad\cd\quad \bigoplus_{Z, W} R(X, Y) \otimes S(Y, Z) \otimes T(Z, W) \otimes U(W, X)
\end{align}
Depending on the semiring $(\bm K, \oplus, \otimes)$, the above query can compute different
things: For the Boolean semiring, it reduces back to the original CQ $\cycle$; for $(\N, +,
\times)$, it can {\em count} the number of cycles for a given $(X, Y)$; and for $(\R, \min,
+)$, it can find the {\em minimum-weight} cycles for a given $(X, Y)$. We recognize two
classes of semirings:

\begin{itemize}[leftmargin=*]
    \item {\em Idempotent semirings}: These are semirings where the $\oplus$ operator is
    {\em idempotent}, meaning that $a \oplus a = a$ for any $a \in \bm K$. These semirings
    include the Boolean semiring, $(\R, \min, +)$, $(\R, \min, \max)$, and $(\R_+, \min,
    \times)$ among others. For all these semirings, the $\panda$ algorithm works just fine
    and achieves the same runtime of $O(N^{\subw(Q, \calS)}\cdot \log N + \OUT)$ using the
    definition of the submodular width from Eq.~\eqref{eq:subw}.
    \item {\em Non-idempotent semirings}: The most notable example is $(\N, +, \times)$,
    which is needed to formulate count queries, e.g., \scqs. For these semirings, the
    $\panda$ algorithm breaks down because the data partitioning used in $\panda$ is {\em
    not} guaranteed in general to produce {\em disjoint} parts, thus we cannot simply count
    by summing up the counts from each part. It remains an open problem whether we can solve
    \scqs (or more generally aggregate queries over non-idempotent semirings) in submodular
    width time $O(N^{\subw(Q, \calS)}\cdot \polylog(N) + \OUT)$. A relaxed version of the
    submodular width, called the {\em sharp-submodular width}, that is tailored to counting
    was introduced in~\cite{10.1145/3426865} along with a matching algorithm. There are
    examples where the sharp-submodular width is strictly larger than the submodular
    width~\cite{DBLP:conf/stoc/BringmannG25}. Moreover, the sharp-submodular width remains
    an unsatisfactory measure for counting since there are queries where counting can be
    done faster\footnote{These faster algorithms are still combinatorial.} than what the
    sharp-submodular width suggests~\cite{DBLP:conf/stoc/BringmannG25}.
\end{itemize}

\subsection{Other Classes of Statistics: $\ell_k$-norms of degree sequences}
\label{sec:extensions:lpnorm}

So far, the statistics $\calS$ we have considered only include {\em degree constraints}.
However, it was shown recently that the polymatroid bound (Theorems~\ref{thm:cq:bound} and
~\ref{thm:ddr:bound}) applies to a much more general class of statistics, called {\em
$\ell_k$-norms constraints}~\cite{2021arXiv211201003V,10.1145/3651597}. The idea of these
statistics is as follows: For simplicity, consider a binary relation $R(X, Y)$, and the
vector $\left(\deg_{R}(Y|X=x)\right)_{x \in \dom^X}$ formed by collecting the degrees of all
values $x$. Suppose we were given an upper bound $N_{Y|X, k}$ on the $\ell_k$-norm of this
vector, for some natural number $k$, i.e., we were told:
\begin{align}
    \left(\sum_{x \in \dom^X} \left(\deg_{R}(Y|X=x)\right)^k\right)^{1/k} \quad\leq\quad N_{Y|X,k}
\end{align}
Then, we can utilize this information to potentially tighten the polymatroid bound
(Eq.~\eqref{eq:entropic:polymatroid:bound}) by adding the following constraint on $h$ as
part of $h \models \calS$:
\begin{align}
    \frac{1}{k} h(X) + h(Y|X) \quad\leq\quad \log_N N_{Y|X,k}
    \label{eq:lk-norm}
\end{align}
Note that the maximum degree $\deg_R(Y|X)$ is simply the $\ell_\infty$-norm of the degree vector,
hence $\ell_k$-norms constraints strictly generalize degree constraints.
In particular, Eq.~\eqref{eq:h:satisfies:dc} is a special case of Eq.~\eqref{eq:lk-norm} when $k = \infty$.

The reason the polymatroid bound still holds is as follows: Let's revisit $\cyclefull$ from
Section~\ref{sec:static:query:plans}, and suppose we were given an upper bound $N_{Y|X,2}$ on the
$\ell_2$-norm of the degree vector of $R(X, Y)$. Consider the probability distribution from
Figure~\ref{fig:4cycle:data}, and specifically the marginal distribution $p_{XY}$ on $XY$,
whose support is contained in $R(X, Y)$. Extend this marginal by creating an identical copy
$Y'$ of $Y$ that is independent of $Y$ given $X$. Note that the support of this distribution
is contained in $R(X, Y) \Join R(X, Y')$, and the size of this join is exactly
$N_{Y|X,2}^2$. Because $Y$ and $Y'$ are independent and identically distributed given $X$,
we get the following, which matches\footnote{To get an exact match, we need to scale the
function $h$ by $1/\log N$, similar to what we did in Section~\ref{sec:static:query:plans}.}
Eq.~\eqref{eq:lk-norm} for $k=2$: {\small
\begin{align}
    h(X) + 2h(Y|X) = h(XYY') \leq \log |R(X, Y) \Join R(X, Y')| \leq 2 \log N_{Y|X,2}
    \label{eq:l2-norm-example}
\end{align}
}

It is natural to extend the definition of the submodular width from Eq.~\eqref{eq:subw}
to include constraints of the form from Eq.~\eqref{eq:lk-norm}.
But now we ask: Can we really extend the $\panda$ algorithm to handle these more general constraints and still achieve the same runtime?
The answer is {\em yes}, and we demonstrate the idea here using the query $\cycle$.
In Section~\ref{subsec:panda:ddr:proof}, we had a cardinality constraint $|R(X, Y)| \leq N$,
which translates to $h(XY) \leq \log_N N$ in the entropy world, which in turn translates to
$p_{XY}(xy) \geq 1/N$ in the probability world.
To satisfy the latter, we defined $p_{XY}(xy) \defeq 1/N$.
Now suppose we have an upper bound on the $\ell_2$-norm of the degree vector $\left(\deg_{R}(Y|X=x)\right)_{x \in \dom^X}$.
In the entropy world, this translates to Eq.~\eqref{eq:l2-norm-example}.
To translate the latter into the probability world, we need to construct sub-probability measures
$p_{X}$ and $p_{Y|X}$ that satisfy:
\begin{align}
    p_{X}(x) \cdot \left(p_{Y|X}(y|x)\right)^2 \quad\geq\quad \frac{1}{N_{Y|X,2}^2}, \quad\text{ for all }(x, y) \in R
    \label{eq:lk-norm:probability}
\end{align}
To achieve this, we define $p_{X}$ and $p_{Y|X}$ as follows:
\begin{align*}
    p_{X}(x) \defeq \frac{\left(\deg_R(Y|X=x)\right)^2}{N_{Y|X,2}^2}, \quad\quad
    p_{Y|X}(y|x) \defeq \frac{1}{\deg_R(Y|X=x)}
\end{align*}
It is straightforward to verify that the above $p_{X}(x)$ and $p_{Y|X}(y|x)$ are valid sub-probability
measures and they satisfy Eq.~\eqref{eq:lk-norm:probability}.
This reasoning generalizes to all $\ell_k$-norm constraints, allowing us to generalize
$\panda$ to handle these constraints.

\subsection{Other Techniques: Fast Matrix Multiplication}
\label{sec:extensions:omega}

For all CQs that have been studied in the literature to date,
the submodular width captures the best known complexity among all combinatorial algorithms.
However, outside combinatorial algorithms, certain queries are known to admit faster algorithms
that use {\em fast matrix multiplication} (FMM).
The complexities of these algorithms are often written in terms of the {\em matrix multiplication exponent} $\omega$, which is the smallest exponent that allows us to multiply two $n \times n$ matrices in time $O(n^{\omega})$.
The best known value of $\omega$ to date is 2.371552~\cite{doi:10.1137/1.9781611977912.134}.
For example, consider the following query, which is the Boolean version of $\cycle$:
\begin{align}
    \cyclebool() &\quad\cd\quad R(X, Y) \wedge S(Y, Z) \wedge T(Z, W) \wedge U(W, X)
\end{align}
Just like $\cycle$, this query has a submodular width of 3/2, under statistics $\cyclestats$ from Eq.~\eqref{eq:4cycle:stats}.
Hence, $\panda$ can only solve this query in time $O(N^{3/2})$.
However, this query admits a faster algorithm that runs in time $O(N^{\frac{4\omega-1}{2\omega+1}})$~\cite{yuster2004detecting, DBLP:journals/siamcomp/DalirrooyfardVW21}.
This begs the question: {\em Can we incorporate FMM into the $\panda$ framework, thus allowing $\panda$ to recover and generalize such algorithms?}

The answer is {\em yes} as was recently shown in~\cite{DBLP:journals/pacmmod/KhamisHS25}.
We summarize the key ideas below.
We start with incorporating FMM into the definition of the submodular width (Eq.~\eqref{eq:subw}).
To that end, the first step is to find an {\em information theoretic interpretation} of the complexity of FMM.
And before that, we revisit the complexity of FMM itself.
Suppose we want to multiply an $(m\times n)$-matrix with an $(n\times p)$-matrix.
Using square FMM, we can partition the two matrices into the largest possible square blocks of size $(q\times q)$,
where $q = \min(m, n, p)$, and then multiply each pair of square blocks using FMM in time $O(q^{\omega})$.
This algorithm takes the following time, where $\gamma\defeq \omega -2$:
\begin{align}
    \max(m \cdot n \cdot p^{\gamma}, \quad m \cdot n^{\gamma} \cdot p, \quad m^{\gamma} \cdot n \cdot p)
\end{align}
In our setting, while trying to evaluate a CQ like $\cyclebool$ above,
our two matrices will be relations like $R(X, Y)$ and $S(Y, Z)$.
In the information theory world, we don't directly have access to the dimensions $m, n, p$.
Instead, we can think of $h(X), h(Y), h(Z)$ as {\em proxies} for $\log m, \log n, \log p$ respectively.
By substituting in the runtime expression above, we get the following {\em information theoretic expression} for the complexity of FMM, on log scale:
\begin{align}
    \mm(X; Y; Z) \quad\defeq\quad \max(h(X)+h(Y)+\gamma h(Z),\quad h(X)+\gamma h(Y)+h(Z),\quad
    \gamma h(X)+h(Y)+h(Z))
    \label{eq:intro:mm}
\end{align}

Now that we have a way to express the FMM complexity in terms of a polymatroid $h$,
the second step to incorporate FMM into the submodular width definition is to find a notion of
(static) query plans that allows us to use FMM.
The original definition from Eq.~\eqref{eq:subw} uses tree decompositions (TDs) as query plans.
For the purpose of capturing FMM, we instead use the language of {\em variable elimination}~\cite{DECHTER199941,10.5555/1622756.1622765,DBLP:conf/pods/KhamisNR16}.
While the two notions are equivalent in the absence of FMM~\cite{DBLP:conf/pods/KhamisNR16},
variable elimination allows us to capture FMM more naturally.
For example, eliminating the variable $Y$ from $\cyclebool$ corresponds to computing the intermediate relation
$P(X, Z) \cd R(X, Y) \wedge S(Y, Z)$.
This intermediate can be computed in two ways:
\begin{itemize}
    \item We can either use a (combinatorial) join algorithm, which costs us $h(XYZ)$ on log scale.
    \item Or we can use FMM, which costs us $\mm(X; Y; Z)$ on log scale, as defined in Eq.~\eqref{eq:intro:mm}.
\end{itemize}
We choose the cheaper of the two options.
Extrapolating from this idea,
there is a natural way to incorporate FMM terms like $\mm(X; Y; Z)$ into the definition of the submodular width,
resulting in what is known as the {\em $\omega$-submodular width}, denoted $\osubw(Q, \calS)$ \cite{DBLP:journals/pacmmod/KhamisHS25}.
Moreover, there is a matching generalization of $\panda$ that can evaluate any Boolean CQ $Q$ in time
$O(N^{\osubw(Q, \calS)}\cdot \polylog(N))$.
For example, $\osubw(\cyclebool,\cyclestats)$ is $\frac{4\omega-1}{2\omega+1}$, thus matching the algorithm
from~\cite{yuster2004detecting, DBLP:journals/siamcomp/DalirrooyfardVW21}.
The $\omega$-submodular width matches the best known runtimes for $k$-cliques~\cite{10.1145/3618260.3649663},
and $k$-cycles~\cite{yuster2004detecting, DBLP:journals/siamcomp/DalirrooyfardVW21}, among others.
We leave the details to~\cite{DBLP:journals/pacmmod/KhamisHS25}.

\section{Open Problems}
\label{sec:open-problems}
We conclude with some major open problems in this domain.

\paragraph*{\bf Count queries (\scq)}
As mentioned earlier in Section~\ref{sec:extensions:counting}, it remains open whether we
can solve \scq queries (i.e., the counting version of CQs) in submodular width time. If this
is not possible, then the next natural questions is: {\em What is the right measure to
capture the complexity of \scqs?} In an attempt to answer this, a counting version of the
submodular width, called the {\em sharp-submodular width}, was proposed
in~\cite{10.1145/3426865}. However, it was shown recently~\cite{DBLP:conf/stoc/BringmannG25}
that this measure is not necessarily the best possible, thus the question remains open.

\paragraph*{\bf Conditional independence constraints}
We have seen earlier in Section~\ref{subsec:output-size-bound} how to translate statistics about
the data into constraints on the polymatroid $h$. However, there are other constraints that
are obeyed by $h$ that do not come from the data, but rather from the {\em query structure}
itself. For example, consider again the uniform probability distribution $p(x, y, z, w)$
over the output of $\cyclefull$ from Figure~\ref{fig:4cycle:data}. Note that this
distribution can be {\em factored} into a product of four functions, matching the four input
relations:
\begin{align}
p(x, y, z, w) \quad=\quad \alpha \cdot f_R(x, y) \cdot f_S(y, z) \cdot f_T(z, w) \cdot f_U(w, x)
\end{align}
where $\alpha$ is a normalizing constant, $f_R(x, y)$ is an {\em indicator function}
indicating whether the pair $(x, y)$ is in $R$, and similarly for $f_S, f_T, f_U$. As is
standard in {\em probabilistic graphical models}, this factorization implies that the two
variables $X$ and $Z$ are {\em conditionally independent} given the two variables $Y$ and
$W$ (and vice versa). In the language of information theory, this means that the {\em
conditional mutual information} $I(X; Z \mid YW)$ is zero, or equivalently:
\begin{align}
    h(XYW) + h(YZW) = h(XYZW) + h(YW)
\end{align}
Note that the above has the same format as submodularity (Eq.~\eqref{eq:h:submodularity}),
except that it is an equality. The above is an extra constraint on $h$ that we can add to
the polymatroid bound (Eq.~\eqref{eq:entropic:polymatroid:bound}) as well as the submodular
width (Eq.~\eqref{eq:subw}). Let's refer to the resulting quantities as the $\ci$-polymatroid bound and the $\ci$-submodular width, respectively, where $\ci$ stands for {\em conditional independence}. This naturally raises two
questions~\cite{kolaitis_et_al:DagRep.12.7.180}:
\begin{itemize}[leftmargin=*]
    \item Are there example queries\footnote{For the specific query $\cyclefull$, we can prove that adding $\ci$ constraints does not improve the bound, but there is no general proof for {\em any} query.} and statistics where the 
    $\ci$-polymatroid bound ($\ci$-submodular width) is strictly better than the polymatroid bound
    (submodular width)?
    \item If the answer to the above is yes, then can we really extend $\panda$ to achieve a
    runtime of the $\ci$-submodular width? This is currently not
    clear since the proof sequence construction from Section~\ref{subsec:proof-sequence}
    does not immediately extend to the new constraints on $h$.
\end{itemize}

\paragraph*{\bf Entropic version of the submodular width}
The submodular width in Eq.~\eqref{eq:subw} is currently defined over polymatroids. It
naturally has an entropic version, where polymatroids are replaced\footnote{We also need to
replace the maximum in Eq.~\eqref{eq:subw} by the supremum since the set of entropic
functions is not closed.} by the subset of entropic
functions~\cite{DBLP:conf/pods/Khamis0S17}. We refer to the resulting quantity as the {\em
entropic width}, denoted $\entw(Q, \calS)$. One way to think of it is that the entropic
width tightens the submodular width by adding extra constraints on $h$ representing {\em
non-Shannon inequalities}~\cite{DBLP:journals/tit/ZhangY97,zhang1998characterization}. It is
not known how to compute the entropic width since characterizing non-Shannon inequalities is
a major open problem. Recent work has shown that the entropic width does indeed provide a
lower bound on the {\em semiring circuit} needed to represent the output of a
CQ~\cite{10.1145/3651588}, thus no {\em ``semiring-based''} algorithm for solving CQs can
improve over the entropic bound by any polynomial factor. This leaves open the question:
{\em Can we really solve CQs in entropic width time?}

\paragraph*{\bf Tractability of computing the polymatroid bound}
It is not known whether computing the polymatroid
bound~\eqref{eq:entropic:polymatroid:bound} is in $\np$. On the positive side, it can be
computed in polynomial time when the input degree constraints are ``simple'' and the query
is a full CQ~\cite{DBLP:journals/pacmmod/ImMNP25}. A closely related question concerns the
length of the shortest proof sequence for a given Shannon-flow inequality: under certain
conditions, this length is polynomial in the number of input
variables~\cite{DBLP:journals/pacmmod/ImMNP25}. This is relevant because the runtime of
$\panda$ has an exponential dependency on the proof sequence length, so bounding this length
is crucial for obtaining practical runtime guarantees. It would be very valuable to identify
broader conditions under which short proof sequences are guaranteed to exist.

\paragraph*{\bf Lower bounds}
As mentioned before, the submodular width achieved by $\panda$ captures the best known
complexity for any CQ over combinatorial algorithms. The question is: {\em Is there a
concrete lower bound we can prove here based on existing complexity conjectures, or do we
need to make a new conjecture about this?} Section~\ref{sec:related:lower-bounds} reviews
some relevant lower bounds. One issue we have to deal with is how to concretely define the
class of {\em ``combinatorial algorithms''}. Prior
works~\cite{DBLP:conf/stoc/BringmannG25,fan_et_al:LIPIcs.ICALP.2023.127} aimed to bypass
this issue in different ways by working under the {\em ``semiring model''} or by
specifically targeting either queries over the $(\min, +)$ semiring or {\em full} queries in order to rule out fast
matrix multiplication (FMM)\footnote{In general, we cannot perform FMM over the $(\min, +)$ semiring because there is no subtraction.
Also, full queries cannot immediately leverage FMM since multiplying two matrices representing a join/projection of two relations $\pi_{X, Z} (R(X, Y)\Join S(Y, Z))$ does not allow us to list $Y$-values afterwards.
See~\cite{DBLP:conf/stoc/BringmannG25} for formal arguments.}. Another possible direction is to try to prove lower bounds for
the FMM-version of the submodular width, known as the {\em $\omega$-submodular
width}~\cite{DBLP:journals/pacmmod/KhamisHS25}, which was discussed earlier in
Section~\ref{sec:extensions:omega}.

\section*{Acknowledgments}
    We thank Dan Olteanu, Eden Chmielewski, Andrei Draghici, Hubie Chen, and Stefan Mengel
    for many discussions and comments that helped us improve the presentation of the $\panda$
    algorithm.
    This work was partially supported by NSF IIS 2314527, NSF SHF 2312195, NSF III 2507117, and a Microsoft professorship.

\bibliographystyle{siam}
\bibliography{bib}

\begin{thebibliography}{10}

\bibitem{DBLP:books/aw/AbiteboulHV95}
{\sc S.~Abiteboul, R.~Hull, and V.~Vianu}, {\em Foundations of Databases}, Addison-Wesley, 1995.

\bibitem{jaguar}
{\sc M.~{Abo Khamis} and H.~{Chen}}, {\em {Jaguar: A Primal Algorithm for Conjunctive Query Evaluation in Submodular-Width Time}}, Proc. {ACM} Manag. Data,  (2026).

\bibitem{10.1145/3426865}
{\sc M.~{Abo Khamis}, R.~R. Curtin, B.~Moseley, H.~Q. Ngo, X.~Nguyen, D.~Olteanu, and M.~Schleich}, {\em Functional aggregate queries with additive inequalities}, ACM Trans. Database Syst., 45 (2020).

\bibitem{10.1145/3651597}
{\sc M.~Abo~Khamis, V.~Nakos, D.~Olteanu, and D.~Suciu}, {\em Join size bounds using lp-norms on degree sequences}, Proc. ACM Manag. Data, 2 (2024).

\bibitem{DBLP:conf/pods/KhamisNR16}
{\sc M.~{Abo Khamis}, H.~Q. Ngo, and A.~Rudra}, {\em {FAQ:} questions asked frequently}, in Proceedings of the 35th {ACM} {PODS} 2016, San Francisco, CA, USA, June 26 - July 01, 2016, {ACM}, 2016, pp.~13--28.

\bibitem{DBLP:conf/pods/KhamisNS16}
{\sc M.~{Abo Khamis}, H.~Q. Ngo, and D.~Suciu}, {\em Computing join queries with functional dependencies}, in Proceedings of the 35th {ACM} {PODS} 2016, San Francisco, CA, USA, June 26 - July 01, 2016, {ACM}, 2016, pp.~327--342.

\bibitem{DBLP:conf/pods/Khamis0S17}
{\sc M.~Abo~Khamis, H.~Q. Ngo, and D.~Suciu}, {\em {What Do Shannon-type Inequalities, Submodular Width, and Disjunctive Datalog Have to Do with One Another?}}, in Proceedings of the 36th {ACM} {PODS} 2017, Chicago, IL, USA, May 14-19, 2017, {ACM}, 2017, pp.~429--444.
\newblock Extended version available at~\url{http://arxiv.org/abs/1612.02503}.

\bibitem{theoretics:13722}
{\sc M.~{Abo Khamis}, H.~Q. Ngo, and D.~Suciu}, {\em {PANDA: Query Evaluation in Submodular Width}}, TheoretiCS, Volume 4 (2025).

\bibitem{pandaexpress}
{\sc M.~{Abo Khamis}, H.~Q. {Ngo}, and D.~{Suciu}}, {\em {PANDAExpress: a Simpler and Faster PANDA Algorithm}}, Proc. {ACM} Manag. Data,  (2026).

\bibitem{MR599482}
{\sc N.~Alon}, {\em On the number of subgraphs of prescribed type of graphs with a given number of edges}, Israel J. Math., 38 (1981), pp.~116--130.

\bibitem{DBLP:journals/algorithmica/AlonYZ97}
{\sc N.~Alon, R.~Yuster, and U.~Zwick}, {\em Finding and counting given length cycles}, Algorithmica, 17 (1997), pp.~209--223.

\bibitem{DBLP:journals/siamcomp/AtseriasGM13}
{\sc A.~Atserias, M.~Grohe, and D.~Marx}, {\em Size bounds and query plans for relational joins}, {SIAM} J. Comput., 42 (2013), pp.~1737--1767.

\bibitem{DBLP:conf/csl/BaganDG07}
{\sc G.~Bagan, A.~Durand, and E.~Grandjean}, {\em On acyclic conjunctive queries and constant delay enumeration}, in Computer Science Logic, 21st International Workshop, {CSL} 2007, 16th Annual Conference of the EACSL, Lausanne, Switzerland, September 11-15, 2007, Proceedings, vol.~4646 of Lecture Notes in Computer Science, Springer, 2007, pp.~208--222.

\bibitem{2025arXiv250320438B}
{\sc C.~{Berkholz} and H.~{Vinall-Smeeth}}, {\em {Factorised Representations of Join Queries: Tight Bounds and a New Dichotomy}}, arXiv e-prints,  (2025), p.~arXiv:2503.20438.

\bibitem{DBLP:conf/stoc/BringmannG25}
{\sc K.~Bringmann and E.~Gorbachev}, {\em A fine-grained classification of subquadratic patterns for subgraph listing and friends}, in Proceedings of the 57th Annual {ACM} Symposium on Theory of Computing, {STOC} 2025, Prague, Czechia, June 23-27, 2025, {ACM}, 2025, pp.~2145--2156.

\bibitem{DBLP:journals/tit/ChanY02}
{\sc T.~H. Chan and R.~W. Yeung}, {\em On a relation between information inequalities and group theory}, {IEEE} Transactions on Information Theory, 48 (2002), pp.~1992--1995.

\bibitem{DBLP:conf/stoc/ChandraM77}
{\sc A.~K. Chandra and P.~M. Merlin}, {\em Optimal implementation of conjunctive queries in relational data bases}, in Proceedings of the 9th Annual {ACM} Symposium on Theory of Computing, May 4-6, 1977, Boulder, Colorado, {USA}, {ACM}, 1977, pp.~77--90.

\bibitem{DBLP:conf/pods/Chaudhuri98}
{\sc S.~Chaudhuri}, {\em An overview of query optimization in relational systems}, in Proceedings of the Seventeenth {ACM} {SIGACT-SIGMOD-SIGART} Symposium on Principles of Database Systems, June 1-3, 1998, Seattle, Washington, {USA}, {ACM} Press, 1998, pp.~34--43.

\bibitem{MR859293}
{\sc F.~R.~K. Chung, R.~L. Graham, P.~Frankl, and J.~B. Shearer}, {\em Some intersection theorems for ordered sets and graphs}, J. Combin. Theory Ser. A, 43 (1986), pp.~23--37.

\bibitem{10.1145/3618260.3649663}
{\sc M.~Dalirrooyfard, S.~Mathialagan, V.~V. Williams, and Y.~Xu}, {\em Towards optimal output-sensitive clique listing or: Listing cliques from smaller cliques}, in Proceedings of the 56th Annual ACM Symposium on Theory of Computing, STOC 2024, New York, NY, USA, 2024, Association for Computing Machinery, p.~923–934.

\bibitem{DBLP:journals/siamcomp/DalirrooyfardVW21}
{\sc M.~Dalirrooyfard, T.~Vuong, and V.~{Vassilevska Williams}}, {\em Graph pattern detection: Hardness for all induced patterns and faster noninduced cycles}, {SIAM} J. Comput., 50 (2021), pp.~1627--1662.

\bibitem{DBLP:journals/ai/Dechter99}
{\sc R.~Dechter}, {\em Bucket elimination: {A} unifying framework for reasoning}, Artif. Intell., 113 (1999), pp.~41--85.

\bibitem{DECHTER199941}
\leavevmode\vrule height 2pt depth -1.6pt width 23pt, {\em Bucket elimination: A unifying framework for reasoning}, Artificial Intelligence, 113 (1999), pp.~41--85.

\bibitem{deeds_et_al:LIPIcs.ICDT.2025.17}
{\sc K.~Deeds and T.~C. Merkl}, {\em {Partition Constraints for Conjunctive Queries: Bounds and Worst-Case Optimal Joins}}, in 28th International Conference on Database Theory (ICDT 2025), vol.~328 of Leibniz International Proceedings in Informatics (LIPIcs), Dagstuhl, Germany, 2025, Schloss Dagstuhl -- Leibniz-Zentrum f{\"u}r Informatik, pp.~17:1--17:18.

\bibitem{10.1145/3695838}
{\sc S.~Deep, H.~Zhao, A.~Z. Fan, and P.~Koutris}, {\em Output-sensitive conjunctive query evaluation}, Proc. ACM Manag. Data, 2 (2024).

\bibitem{DBLP:journals/ftdb/DingNC24}
{\sc B.~Ding, V.~R. Narasayya, and S.~Chaudhuri}, {\em Extensible query optimizers in practice}, Found. Trends Databases, 14 (2024), pp.~186--402.

\bibitem{10.1145/2448496.2448508}
{\sc A.~Durand and S.~Mengel}, {\em Structural tractability of counting of solutions to conjunctive queries}, in Proceedings of the 16th International Conference on Database Theory, ICDT '13, New York, NY, USA, 2013, Association for Computing Machinery, p.~81–92.

\bibitem{fan_et_al:LIPIcs.ICALP.2023.127}
{\sc A.~Z. Fan, P.~Koutris, and H.~Zhao}, {\em {The Fine-Grained Complexity of Boolean Conjunctive Queries and Sum-Product Problems}}, in 50th International Colloquium on Automata, Languages, and Programming (ICALP 2023), vol.~261 of Leibniz International Proceedings in Informatics (LIPIcs), Dagstuhl, Germany, 2023, Schloss Dagstuhl -- Leibniz-Zentrum f{\"u}r Informatik, pp.~127:1--127:20.

\bibitem{10.1145/3651588}
\leavevmode\vrule height 2pt depth -1.6pt width 23pt, {\em Tight bounds of circuits for sum-product queries}, Proc. ACM Manag. Data, 2 (2024).

\bibitem{MR2104047}
{\sc E.~Friedgut}, {\em Hypergraphs, entropy, and inequalities}, Amer. Math. Monthly, 111 (2004), pp.~749--760.

\bibitem{MR1639767}
{\sc E.~Friedgut and J.~Kahn}, {\em On the number of copies of one hypergraph in another}, Israel J. Math., 105 (1998), pp.~251--256.

\bibitem{friedgut-kahn-1998}
\leavevmode\vrule height 2pt depth -1.6pt width 23pt, {\em On the number of copies of one hypergraph in another}, Israel Journal of Mathematics, 105 (1998), pp.~251--256.

\bibitem{DBLP:conf/pods/GottlobGLS16}
{\sc G.~Gottlob, G.~Greco, N.~Leone, and F.~Scarcello}, {\em Hypertree decompositions: Questions and answers}, in Proceedings of the 35th {ACM} {SIGMOD-SIGACT-SIGAI} Symposium on Principles of Database Systems, {PODS} 2016, San Francisco, CA, USA, June 26 - July 01, 2016, {ACM}, 2016, pp.~57--74.

\bibitem{DBLP:conf/pods/GottlobLV09}
{\sc G.~Gottlob, S.~T. Lee, and G.~Valiant}, {\em Size and treewidth bounds for conjunctive queries}, in Proceedings of the Twenty-Eigth {ACM} {SIGMOD-SIGACT-SIGART} Symposium on Principles of Database Systems, {PODS} 2009, June 19 - July 1, 2009, Providence, Rhode Island, {USA}, {ACM}, 2009, pp.~45--54.

\bibitem{DBLP:journals/jacm/GottlobLVV12}
{\sc G.~Gottlob, S.~T. Lee, G.~Valiant, and P.~Valiant}, {\em Size and treewidth bounds for conjunctive queries}, J. {ACM}, 59 (2012), pp.~16:1--16:35.

\bibitem{GOTTLOB2002579}
{\sc G.~Gottlob, N.~Leone, and F.~Scarcello}, {\em Hypertree decompositions and tractable queries}, Journal of Computer and System Sciences, 64 (2002), pp.~579--627.

\bibitem{provenance-semirings}
{\sc T.~J. Green, G.~Karvounarakis, and V.~Tannen}, {\em Provenance semirings}, in Proceedings of the Twenty-Sixth ACM SIGMOD-SIGACT-SIGART Symposium on Principles of Database Systems, PODS '07, New York, NY, USA, 2007, Association for Computing Machinery, p.~31–40.

\bibitem{DBLP:conf/pods/GreenT17}
{\sc T.~J. Green and V.~Tannen}, {\em The semiring framework for database provenance}, in Proceedings of the 36th {ACM} {SIGMOD-SIGACT-SIGAI} Symposium on Principles of Database Systems, {PODS} 2017, Chicago, IL, USA, May 14-19, 2017, {ACM}, 2017, pp.~93--99.

\bibitem{DBLP:conf/soda/GroheM06}
{\sc M.~Grohe and D.~Marx}, {\em Constraint solving via fractional edge covers}, in Proceedings of the Seventeenth Annual {ACM-SIAM} Symposium on Discrete Algorithms, {SODA} 2006, Miami, Florida, USA, January 22-26, 2006, {ACM} Press, 2006, pp.~289--298.

\bibitem{DBLP:journals/talg/GroheM14}
\leavevmode\vrule height 2pt depth -1.6pt width 23pt, {\em Constraint solving via fractional edge covers}, {ACM} Trans. Algorithms, 11 (2014), pp.~4:1--4:20.

\bibitem{10.1145/3725241}
{\sc X.~Hu}, {\em Output-optimal algorithms for join-aggregate queries}, Proc. ACM Manag. Data, 3 (2025).

\bibitem{DBLP:journals/pacmmod/ImMNP25}
{\sc S.~Im, B.~Moseley, H.~Q. Ngo, and K.~Pruhs}, {\em Efficient algorithms for cardinality estimation and conjunctive query evaluation with simple degree constraints}, Proc. {ACM} Manag. Data, 3 (2025), pp.~96:1--96:26.

\bibitem{DBLP:conf/stoc/Itai77}
{\sc A.~Itai and M.~Rodeh}, {\em Finding a minimum circuit in a graph}, in Proceedings of the 9th Annual {ACM} Symposium on Theory of Computing, May 4-6, 1977, Boulder, Colorado, {USA}, {ACM}, 1977, pp.~1--10.

\bibitem{DBLP:journals/pacmmod/KhamisHS25}
{\sc M.~A. Khamis, X.~Hu, and D.~Suciu}, {\em Fast matrix multiplication meets the submodular width}, Proc. {ACM} Manag. Data, 3 (2025), pp.~98:1--98:26.

\bibitem{10.1145/2967101}
{\sc M.~A. Khamis, H.~Q. Ngo, C.~R\'{e}, and A.~Rudra}, {\em Joins via geometric resolutions: Worst case and beyond}, ACM Trans. Database Syst., 41 (2016).

\bibitem{kolaitis_et_al:DagRep.12.7.180}
{\sc P.~G. Kolaitis, A.~E. Romashchenko, M.~Studen\'{y}, D.~Suciu, and T.~A. Boege}, {\em {Algorithmic Aspects of Information Theory (Dagstuhl Seminar 22301)}}, Dagstuhl Reports, 12 (2023), pp.~180--204.

\bibitem{DBLP:conf/pods/KolaitisV98}
{\sc P.~G. Kolaitis and M.~Y. Vardi}, {\em Conjunctive-query containment and constraint satisfaction}, in Proceedings of the Seventeenth {ACM} {SIGACT-SIGMOD-SIGART} Symposium on Principles of Database Systems, June 1-3, 1998, Seattle, Washington, {USA}, {ACM} Press, 1998, pp.~205--213.

\bibitem{DBLP:books/daglib/0023091}
{\sc D.~Koller and N.~Friedman}, {\em Probabilistic Graphical Models - Principles and Techniques}, {MIT} Press, 2009.

\bibitem{DBLP:journals/pvldb/LeisGMBK015}
{\sc V.~Leis, A.~Gubichev, A.~Mirchev, P.~A. Boncz, A.~Kemper, and T.~Neumann}, {\em How good are query optimizers, really?}, Proc. {VLDB} Endow., 9 (2015), pp.~204--215.

\bibitem{Maier:1983:TRD:1097039}
{\sc D.~Maier}, {\em Theory of Relational Databases}, Computer Science Pr, 1983.

\bibitem{DBLP:journals/jacm/Marx13}
{\sc D.~Marx}, {\em {Tractable Hypergraph Properties for Constraint Satisfaction and Conjunctive Queries}}, J. {ACM}, 60 (2013), pp.~42:1--42:51.

\bibitem{DBLP:conf/pods/NgoPRR12}
{\sc H.~Q. Ngo, E.~Porat, C.~R{\'{e}}, and A.~Rudra}, {\em Worst-case optimal join algorithms: [extended abstract]}, in Proceedings of the 31st {ACM} {SIGMOD-SIGACT-SIGART} Symposium on Principles of Database Systems, {PODS} 2012, Scottsdale, AZ, USA, May 20-24, 2012, {ACM}, 2012, pp.~37--48.

\bibitem{10.1145/3180143}
{\sc H.~Q. Ngo, E.~Porat, C.~R\'{e}, and A.~Rudra}, {\em Worst-case optimal join algorithms}, J. ACM, 65 (2018).

\bibitem{DBLP:journals/sigmod/NgoRR13}
{\sc H.~Q. Ngo, C.~R{\'{e}}, and A.~Rudra}, {\em Skew strikes back: new developments in the theory of join algorithms}, {SIGMOD} Rec., 42 (2013), pp.~5--16.

\bibitem{DBLP:books/daglib/0011128}
{\sc R.~Ramakrishnan and J.~Gehrke}, {\em Database management systems {(3.} ed.)}, McGraw-Hill, 2003.

\bibitem{DBLP:conf/icdt/Veldhuizen14}
{\sc T.~L. Veldhuizen}, {\em Triejoin: {A} simple, worst-case optimal join algorithm}, in Proc. 17th International Conference on Database Theory (ICDT), Athens, Greece, March 24-28, 2014, OpenProceedings.org, 2014, pp.~96--106.

\bibitem{2021arXiv211201003V}
{\sc S.~{Vikneshwar Mani Jayaraman}, C.~{Ropell}, and A.~{Rudra}}, {\em {Worst-case Optimal Binary Join Algorithms under General $\ell_p$ Constraints}}, arXiv e-prints,  (2021), p.~arXiv:2112.01003.

\bibitem{doi:10.1137/1.9781611977912.134}
{\sc V.~V. Williams, Y.~Xu, Z.~Xu, and R.~Zhou}, {\em New Bounds for Matrix Multiplication: from Alpha to Omega}, pp.~3792--3835.

\bibitem{DBLP:conf/vldb/Yannakakis81}
{\sc M.~Yannakakis}, {\em Algorithms for acyclic database schemes}, in Very Large Data Bases, 7th International Conference, September 9-11, 1981, Cannes, France, Proceedings, {IEEE} Computer Society, 1981, pp.~82--94.

\bibitem{yuster2004detecting}
{\sc R.~Yuster and U.~Zwick}, {\em Detecting short directed cycles using rectangular matrix multiplication and dynamic programming.}, in SODA, vol.~4, 2004, pp.~254--260.

\bibitem{10.5555/1622756.1622765}
{\sc N.~L. Zhang and D.~Poole}, {\em Exploiting causal independence in bayesian network inference}, J. Artif. Int. Res., 5 (1996), p.~301–328.

\bibitem{zhang1994simple}
{\sc N.~X. Zhang and D.~Poole}, {\em A simple approach to bayesian network computations}, in Proceedings of the Canadian AI Conference, Springer, 1994, pp.~207--214.

\bibitem{DBLP:journals/tit/ZhangY97}
{\sc Z.~Zhang and R.~W. Yeung}, {\em A non-shannon-type conditional inequality of information quantities}, {IEEE} Trans. Information Theory, 43 (1997), pp.~1982--1986.

\bibitem{zhang1998characterization}
{\sc Z.~Zhang and R.~W. Yeung}, {\em On characterization of entropy function via information inequalities}, IEEE Transactions on Information Theory, 44 (1998), pp.~1440--1452.

\end{thebibliography}


\end{document}